\documentclass[aps,prl,twocolumn,showpacs,superscriptaddress,longbibliography]{revtex4-2}

\usepackage{amsmath,amssymb,amsthm}
\usepackage{braket}
\usepackage{svg}
\usepackage{graphicx}
\usepackage{dcolumn}
\usepackage{bm}
\usepackage{balance}
\usepackage{physics}
\usepackage{xcolor}
\usepackage{booktabs}
\usepackage{hyperref}
\hypersetup{
    colorlinks,
    linkcolor={blue},
    citecolor={blue},
    urlcolor={blue}
}

\definecolor{mscolor}{rgb}{0,0.5,0.5}

\newcommand{\us}{\ \mu\rm{s}}
\newcommand{\um}{\ \mu\rm{m}}

\newcommand{\uwmphysics}{Department of Physics, University of Wisconsin-Madison, Madison, Wisconsin 53706, USA}
\newcommand{\uwmchemistry}{Department of Chemistry, University of Wisconsin-Madison, Madison, Wisconsin 53706, USA}
\newcommand{\INFQmadison}{Infleqtion, Madison, Wisconsin 53703, USA}

\newcommand {\rsub}[1]{\textcolor{black}{#1}}
\newcommand{\rsubmath}[1]{{\color{black}#1}}

\begin{document}

\title{Entangling gate performance and fidelity limits with neutral atom F\"orster resonances}

\author{Sam A. Norrell}
\thanks{These authors contributed equally to this work.}
\affiliation{\uwmphysics}

\author{Yufei Shen}
\thanks{These authors contributed equally to this work.}
\affiliation{\uwmphysics}

\author{Mark Saffman}
\affiliation{\uwmphysics}
\affiliation{\INFQmadison}

\author{Matthew Otten}
\affiliation{\uwmphysics}
\affiliation{\uwmchemistry}

\date{\today}

\begin{abstract}
Neutral-atom entangling gates are commonly analyzed with a single effective Rydberg-pair state, but near F\"orster resonances the pair manifold contains resonantly coupled interaction channels that change both the control landscape and the achievable fidelity.
We develop a two-eigenstate model for this regime and show that when allowing for coupling to both pair states in the resonance, the gate fidelity is bounded by $\mathcal{F}\leq 1-(\pi/2)/(V\tau_R)$, for interaction strength $V$ and Rydberg lifetime $\tau_R$.
We construct a gate protocol that saturates this bound in the large-Rabi-frequency limit, improving the existing fidelity limit by approximately $40\%$.
We also evaluate common gate protocols near F\"orster resonances and find that retaining the exchange dynamics increases predicted fidelities by up to two orders of magnitude over earlier treatments.
\end{abstract}

\maketitle
\textit{Introduction} -- Neutral-atom quantum processors based on Rydberg interactions have matured rapidly, with two-qubit gate fidelities now exceeding fault-tolerance
thresholds \cite{Evered2023, Peper2025, Muniz2025, Tsai2025, Radnaev2025}.
A key challenge in pushing fidelity further is understanding and controlling the two fundamental error channels: decay from Rydberg states during the gate, and coherent leakage into off-target states in the Rydberg manifold due to finite blockade \cite{saffman2005a, Doultsinos2025, Jandura2022, Jandura2023, Pagano2022, Mohan2023, Cole2026, Petrosyan2017}.
To model these error sources, most analyses adopt a single-channel interaction model in which an effective state $\ket{rr}$ describes and condenses all Rydberg interactions \cite{Walker2008, Jandura2022, Jaksch2000, XFShi2018}. Real Rydberg systems, however, support multiple interaction channels that qualitatively change the control landscape.

In this letter, we develop a systematic treatment of \emph{two-state} pair physics and translate the resulting mechanism into concrete pulse sequences for entangling gates.
Two-state models are of particular relevance, as resonant dipole-dipole interactions between Rydberg atoms (so-called F\"orster resonances) result in symmetric, bifurcated interaction channels. 
F\"orster resonances are popular state choices in quantum processors as they increase gate fidelities at longer interaction range due to  $1/r^3$ scaling, compared to the $1/r^6$ Van der Waals scaling of induced dipole interactions \cite{Ravets2014, Ryabtsev2010, Nipper2012, Ravets2015, Beterov2015, Beterov2016b, Browaeys2016, Anand2024, Giudici2025, Palm2026}.
Whereas F\"orster resonances increase atom-atom interaction magnitude $V$, we compare such resonances to single-channel models with the same effective energy, thereby isolating the effects of the increased dimensionality alone.

In F\"orster resonances, the two equally spaced eigenstates combine under \rsub{a single} optical drive into a laser-\emph{bright} and a laser-\emph{dark} state \cite{Petrosyan2017,Mostaan2026, Keating2015}.
The dark state is not directly driven but is coupled to the bright state via the pair interaction. This bright / dark interaction gives rise to a cancellation mechanism that can strongly suppress coherent gate errors. We investigate the dynamics between the bright and dark states, and demonstrate the importance of including this interaction  when evaluating gate performance for F\"orster resonant states.

\textit{Atomic Models} --
We consider a pair of atoms -- $A$ and $B$ -- in Rydberg states $\ket{a}$, $\ket{b}$, respectively, separated by $\textbf{r}=r\hat{n}$ (as in Fig. \ref{fig:energy-levels}(a)).
The atoms interact via a dipole-dipole interaction, summed over many Rydberg states $\ket{k},\ket{l},\ket{m},\ket{n}$,
\begin{multline}
    H_{\rm dd}(\textbf{r}) =\frac{1}{4\pi\epsilon_0r^3}\sum_{k,l,m,n}\ket{k}_{AA}\hspace{-.07cm}\bra{l}\otimes \ket{m}_{BB}\hspace{-.07cm}\bra{n} \times \\ \left[\textbf{d}_{kl,A}\cdot \textbf{d}_{mn,B}-3\left( \textbf{d}_{kl,A}\cdot \hat{n}\right)\left(\textbf{d}_{mn,B}\cdot \hat{n}\right)\right],
    \label{eq:Hdd}
\end{multline}
for dipole matrix elements $\textbf{d}_{ij,X}=\bra{i}\textbf{d}\ket{j}_X$.
In general, the eigenstates $\ket{e_i}$ of $H_{\rm dd}$ (with energy eigenvalues $V_i$) are admixtures of single-atom Rydberg states, explicitly $\ket{e_i} =\sum_{p,q}c_{pq}\ket{r_p, r_q}$. We define the overlap to the target Rydberg-Rydberg pair state $o_i\equiv |\langle e_i | ab\rangle|^2$.

\begin{figure*}
    \centering
    \includegraphics[width=\linewidth]{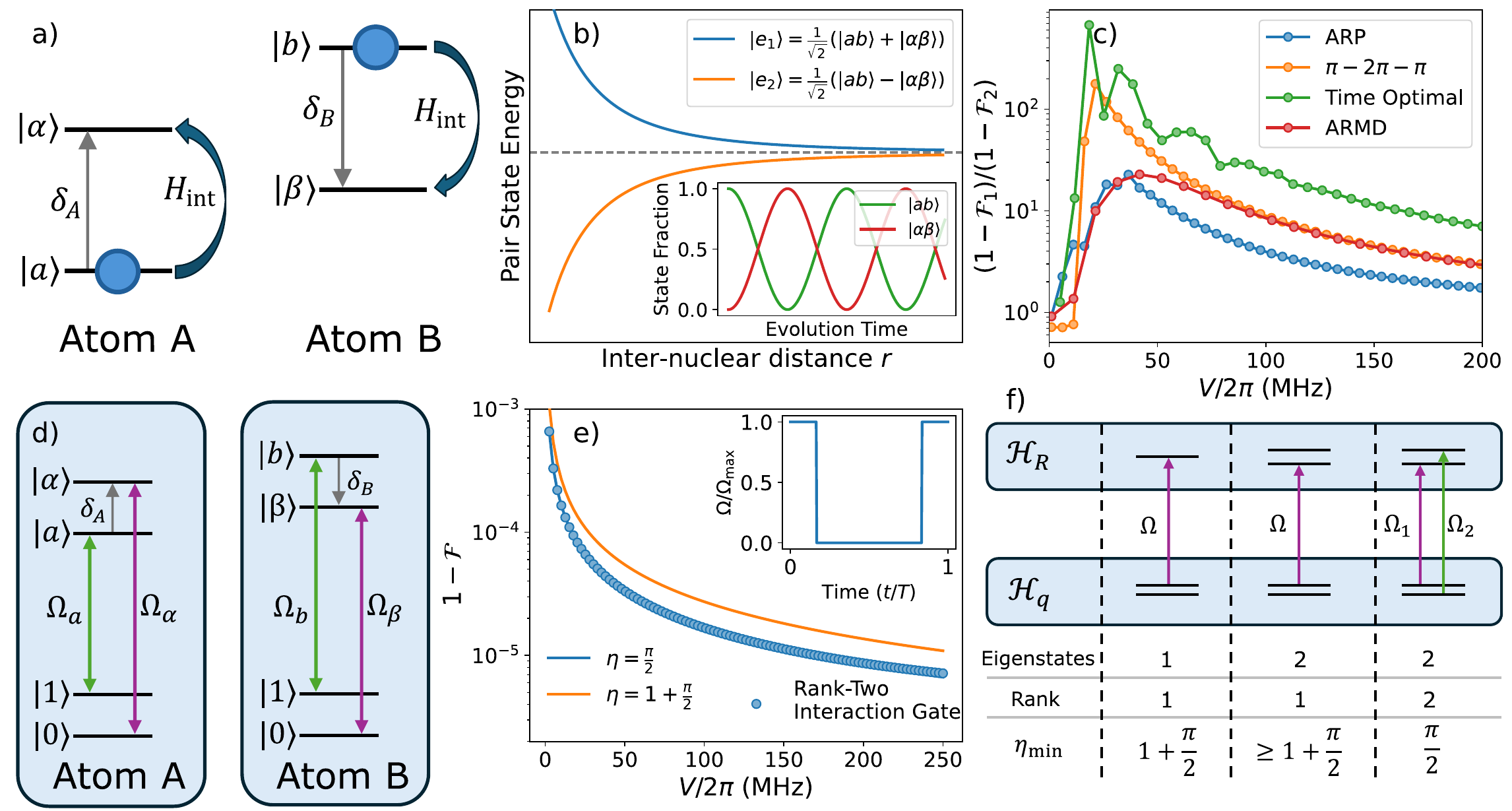}
    \caption{ 
    (a) Atomic structure of a F\"orster resonance. When $\delta_A=-\delta_B$, the resonant exchange process $\ket{a}\leftrightarrow\ket{\alpha}$ and $\ket{b}\leftrightarrow\ket{\beta}$ occurs.
    (b) A F\"orster resonance shown in the pair basis with two interaction eigenstates $\ket{e_1}$ and $\ket{e_2}$. In the inset, coherent oscillations between $\ket{ab}$ and $\ket{\alpha \beta}$ are shown. 
    (c) Gate fidelity simulations under the one- and two-eigenstate models of a F\"orster resonance. Fidelity calculations differ by up to two orders of magnitude when using atomic models that include the F\"orster exchange interaction. The subscripts in the axis label describe the number of eigenstates used to evaluate gate fidelity. \rsub{Four separate gate protocols are analyzed.}
    (d) Depiction of a rank-two coupling. Here, each atom is coupled to two Rydberg states, which then evolve under $H_{\rm int}$, producing entanglement.
    (e) A rank-two interaction gate saturates the bound of $\eta=\frac{\pi}{2}$.
    \rsub{(f) Summary of the bound coefficients $\eta$ for rank-one and rank-two drives under the one- and two-eigenstate models.}}
    \label{fig:energy-levels}
\end{figure*}

In standard treatments, all interaction channels $i$ are condensed into a single effective state $\ket{rr}$, which is shifted by an amount  $V\equiv \left(\sum_i\frac{o_i}{V_i^2}\right)^{-1/2}$ \cite{Walker2008}.
We refer to such a model as the ``one-eigenstate'' model.
In the limit of a perfect F\"orster resonance, $H_{\rm dd}$ has exactly two eigenstates, $\ket{e_1}$ and $\ket{e_2}$, shifted symmetrically by energies $V_1=-V_2$ from the non-interacting target state energy. We refer to this model as the ``two-eigenstate'' model. Of note, the one-eigenstate model applied to a F\"orster resonance condenses the two channels to $\ket{rr}$ such that $|V_1|=|V_2|=V$, but as we show below, this operation does not preserve state dynamics crucial to gate fidelity calculations. 

A F\"orster resonance requires two additional nearby states $\ket{\alpha}_A,\ket{\beta}_B$ with energy spacings $\delta_A=E_\alpha-E_a$ and $\delta_B=E_\beta-E_b$  satisfying $\delta_A+\delta_B=0$. In the resonance, a virtual photon process causes coherent population exchange $\ket{ab}\leftrightarrow\ket{\alpha\beta}$. Figure \ref{fig:energy-levels}\rsub{(a-b)} shows a F\"orster resonance in both the single-atom and pair-state frames, along with the population transfer between $\ket{ab}$ and $\ket{\alpha \beta}$. The remaining panels summarize our core findings. Here, and in all subsequent figures, $\tau_R=150 \us$, corresponding to $n\approx70$ for rubidium or cesium $ns$  states at $300$ K.

For atom $X\in\{A, B\}$, consider the distinct Hilbert spaces $\mathcal{H}_q^X$ and $\mathcal{H}_R^X$, describing the qubit and Rydberg manifolds respectively. We define the \textit{rank} of any pulse sequence as the maximum of $\{\textsf{rank}(H_{\rm drive}^A),\textsf{rank}(H_{\rm drive}^B)\}$, where $H_{\rm drive}^X : \mathcal{H}_q^X\mapsto \mathcal{H}_R^X$.
Typically, neutral-atom Rydberg gates use rank-one pulses. Rank-two pulses have been proposed, enabling fast or long-range gates \cite{Bergonzoni2025, Ildefonso2025,Giudici2025}.

Within the Rydberg manifold of a F\"orster resonance, the non-zero elements of $H_{\rm dd}$ require both atoms to change states, as $H_{\rm dd}$ only contains products of single-atom operators. This results in the interaction Hamiltonian $(\hbar=1)$, in the basis $\{\ket{ab},\ket{\alpha\beta},\ket{a\beta},\ket{\alpha b}\}$ of
\begin{equation}
    H_{\rm int} = \begin{pmatrix}
        0 & V & 0 & 0 \\
        V^* & 0 & 0 & 0 \\
        0 & 0 & \delta_B & W \\
        0 & 0 & W^* & \delta_A
    \end{pmatrix},
\end{equation}
where $V$ and $W$ are calculated per Eq.~(\ref{eq:Hdd}). For atoms aligned along a shared quantization axis, $V=W$. We assume this geometry henceforth. In the rotating frame, and by making the rotating wave approximation (RWA), the Hamiltonian reduces to
\begin{equation}
    H_{\rm int} = V\left( \ket{ab}\bra{\alpha\beta} + \ket{\alpha\beta}\bra{ab}\right).
    \label{eq:Hint}
\end{equation}
The validity of RWA mandates $V/|\delta_X| \ll 1$. For Rydberg level spacings of order GHz and long-range F\"orster resonances of order MHz, this condition is easily satisfied. \rsub{Nearly perfect F\"orster resonances realizing $H_{\rm int}$ as in Eq. (3) appear naturally in real atomic systems, an example of which we present in the Supplemental \cite{Norrell2026SM}.}

\textit{Existing Fidelity Bound} -- Following existing gate fidelity analysis frameworks \cite{Doultsinos2025b,Wesenberg2007}, we define for any two-atom state $\ket{\psi}$ the Rydberg population $P_r = \bra{\psi}\hat{\Pi}_r\ket{\psi}$, where $\hat{\Pi}_r$ is the projector onto all configurations in which at least one atom occupies a Rydberg state.
For a gate of duration $T\ll \tau_R=1/\Gamma$, the decay error is
\begin{equation}
    \varepsilon_{\rm decay}=\Gamma\int_0^T dt\, P_r(t)=\Gamma T_R,
\end{equation}
where $T_R$ is the integrated Rydberg population time. Minimizing gate errors due to finite Rydberg lifetime is equivalent to minimization of $T_R$.

The minimum integrated Rydberg time to generate complete entanglement is measured by the entropy $S=-\log_2(c_1^2)$, where $c_1$ is the largest Schmidt coefficient of the atom pair.
It can be shown that \cite{Doultsinos2025b}
\begin{equation}
  T_R \geq \int_0^1 ds\,  G(s), \quad
  G(s) = \min_{\psi:\,S(\psi)=s}
         \frac{P_r(\psi)}{|\dot{S}(\psi)|},
  \label{eq:eta}
\end{equation}
where $G$ is minimized over all two-atom states with minimum entropy $s$. The coefficient
\begin{equation}
    \eta \equiv V\int_0^1 ds\, G(s)
\end{equation}
is therefore the figure-of-merit: gates that prepare a maximally entangled state satisfy $1-\mathcal{F}\geq \frac{\eta}{V\tau_R}$. It \rsub{has been proven} that for  \rsub{one-eigenstate systems with rank-one excitation}
\begin{equation}
    \eta \geq 1 + \frac{\pi}{2},
\end{equation}
and this bound is tight \cite{Doultsinos2025b}.
\rsub{We show in the Supplemental that this same bound holds for two-eigenstate models with rank-one access \cite{Norrell2026SM}.}

\textit{Rank-Two Gate Fidelity Bound} --
\begin{figure}
    \centering
    \includegraphics[width=\linewidth]{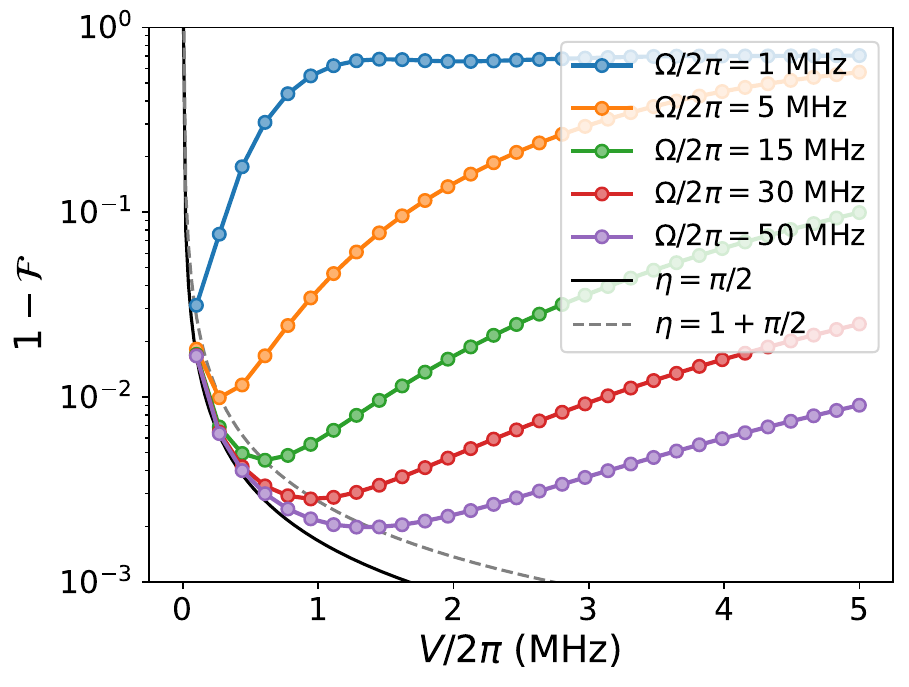}
    \caption{Gate performance of rank-two interaction gate at finite Rabi frequency, breaking the rank-one $\eta=1+\frac{\pi}{2}$ bound. Pulses saturate the bound $\eta=\frac{\pi}{2}$ identically for $\Omega\rightarrow\infty$.
    }
    \label{fig:rank-2-gate}
\end{figure}
We consider now a gate protocol which directly addresses each of the single-atom states of a F\"orster resonance, $\ket{a},\ket{b},\ket{\alpha}, \ket{\beta}$. In this model, the fundamental gate limit is set by
\begin{equation}
  \eta \geq \frac{\pi}{2}.
  \label{eq:etabound}
\end{equation}
The full proof of this bound is presented in the Supplemental Material \cite{Norrell2026SM}. This value is approximately 40\% lower than the value for \rsub{rank-one gates, and directly attributes the increase in fidelity to the increase of the rank of the excitation pulse.}

To saturate the equality of $\tfrac{\pi}{2}$ bound, we consider a $\pi-\rm{gap}-\pi$ sequence, with rank-two $\pi$ pulses of Rabi frequency $\Omega$ on atoms $A$ and $B$ (shown in Fig.~\ref{fig:energy-levels}(\rsub{d})):
\begin{align*}
\text{Drive on $A$: } & \ket{0}_A \leftrightarrow \ket{\alpha}_A,\ 
\ket{1}_A \leftrightarrow \ket{a}_A\\
\text{Drive on $B$: } &\ket{0}_B \leftrightarrow \ket{\rsub{\beta}}_B, \
\ket{1}_B \leftrightarrow \ket{\rsub{b}}_B.
\end{align*}
This gate is analogous to other  Rydberg interaction gate protocols \cite{Jaksch2000, YChew2022, Petrosyan2017}, but extended to rank-two drive. All subsequent analysis is performed in the rotating frame, so Eq.~(\ref{eq:Hint}) is the relevant interaction. We address the three steps of the pulse individually.

\textbf{(i) Excitation ($\pi$ Pulse 1)} The Rydberg population for each atom is given by
\begin{equation}
    p_r(t)=\frac{\Omega^2}{\Omega^2 + V^2}\sin^2\left(\frac{\sqrt{\Omega^2 + V^2}}{2}t\right).
\end{equation}
The off-resonant cross-terms, $\ket{0}\to\ket{a}$ for instance, are forbidden by selection rules, as $\ket{a}$ and $\ket{\alpha}$ are assumed to have different parities  $(-1)^\ell$ with $\ell$ the orbital angular momentum. The total Rydberg population during the excitation pulse is
\begin{equation}
    P_r(t) = p_r^A(t)+p_r^B(t)=2p_r(t).
\end{equation}
Setting $t_\pi=\pi/\Omega$,
\begin{equation}
    \eta_{\rm excitation} = V\int_0^{t_\pi}dt\, P_r(t) = \pi\frac{V}{\Omega}+\mathcal{O}\left[\left(\frac{V}{\Omega} \right)^3\right].
\end{equation}
There is an additional error term, namely the incomplete population transfer, which scales as $\left(V/\Omega\right)^2$. Assuming $\Omega\gg V$, this error is sub-leading and the following map is realized:
\begin{align*}
    &\ket{00}\to -\ket{\rsub{\alpha \beta}}, \quad \ket{01}\to -\ket{\rsub{\alpha b}},\\ 
    &\ket{10}\to -\ket{\rsub{a\beta}}, \quad \ket{11}\to -\ket{\rsub{a b}}.
\end{align*}

\textbf{(ii) Interaction (``gap'')}
The states $\ket{a\beta}$ and $\ket{\alpha b}$ are eigenstates of $H_{\rm int}$ and are thus stationary in the rotating frame. The states $\ket{ab}$ and $\ket{\alpha\beta}$ oscillate under the F\"orster coupling. Explicitly, for $t_{\rm gap}=\pi/(4V)$:
\begin{align*}
    -\ket{ab}&\to \tfrac{1}{\sqrt{2}}\left(-\ket{ab}+i\ket{\alpha\beta}\right),\\ 
    -\ket{\alpha\beta}&\to \tfrac{1}{\sqrt{2}}\left(-\ket{\alpha\beta}+i\ket{ab}\right),\\
    -\ket{a\beta}&\to -\ket{a\beta}, \\
    -\ket{\alpha b}&\to -\ket{\alpha b}.
\end{align*}
The total Rydberg population is $P_r(t) = 2$ for all four inputs throughout the gap, so
\begin{equation}
    \eta_{\rm gap} = V \int_0^{t_{\rm gap}} dt\, P_r(t)
     = V \cdot 2 \cdot \frac{\pi}{4V}
     = \frac{\pi}{2}.
\end{equation}

\textbf{(iii) Return ($\pi$ Pulse 2)}
The system is brought back to the qubit manifold via the same $\pi$ pulse as in the excitation step, so the calculation of $\eta$ is identical. This protocol realizes a $\sqrt{i\sf{SWAP}}^\dagger$ gate, which is maximally entangling.

Summing the three contributions to $\eta$,
\begin{equation}
    \eta=\eta_{\rm excitation} + \eta_{ \rm gap} + \eta_{\rm return}= \frac{\pi}{2} + 2\pi\frac{V}{\Omega}+\mathcal{O}\left[\left(\frac{V}{\Omega} \right)^3\right],
\end{equation}
so that in the limit $\Omega\gg V$, $\eta\to\pi/2$.

\begin{figure*}[!t]
    \centering
    \includegraphics[width=\linewidth]{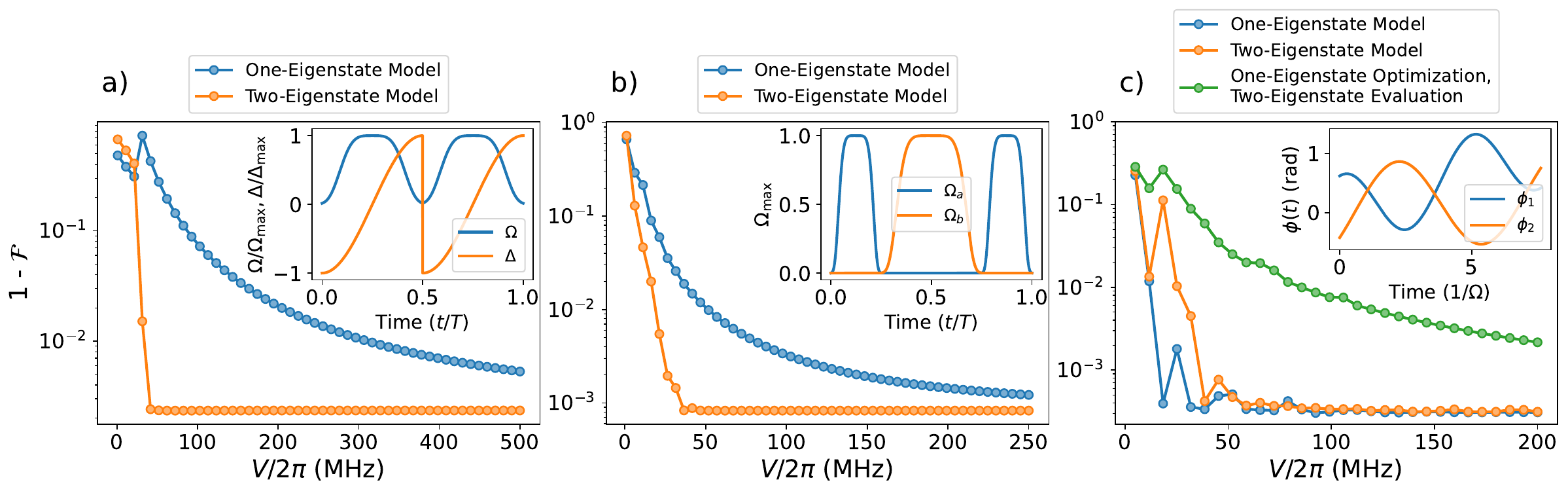}
    \caption{Comparison of gate fidelities evaluated under one- or two-eigenstate models, with pulse sequences shown in the insets. (a) ARP gate. (b) $\pi-2\pi-\pi$ gate (c) TO gate. Gate pulse sequences are optimized and evaluated for one- and two-eigenstate models. The pulses optimized under the one-eigenstate model are also evaluated under the two-eigenstate model, shown in green. Optimized pulses under each model at $V\approx2\pi\times 100$ MHz are shown in the inset.
    }
    \label{fig:gates_1_vs_2}
\end{figure*}

Concrete examples of the performance of this pulse sequence are shown in Fig.~\ref{fig:rank-2-gate}. We find this gate provides fidelities closely approaching equality in Eq.~(\ref{eq:etabound}), with complete saturation in the limit $\Omega\to\infty$. Finite Rabi frequency causes incomplete Rydberg transfer, leading to gate infidelity above the bound. This error channel can be mitigated by pulse sequences robust to off-resonant drive, resulting in higher gate performance at lower Rabi frequency.
We emphasize, though, that coupling to multiple Rydberg states adds experimental overhead when compared to rank-one architectures. We address this overhead and provide pathways to achieving this in the Supplemental Material \cite{Norrell2026SM}.

\textit{Increased Gate Fidelity with F\"orster Resonances} -- 
Owing to the experimental challenges in realizing rank-two access to the Rydberg manifold, we now consider rank-one implementations, where population is only driven to $\ket{a}$ and $\ket{b}$.

We analyze the following gate protocols: Adiabatic Rapid Passage (ARP) \cite{Saffman2020}, $\pi-2\pi-\pi$ \cite{Jaksch2000}, Time Optimal (TO) \cite{Jandura2022}\rsub{, and Almost Resonant Modulated Driving (ARMD) \cite{fan2024fast}}. We calculate gate fidelities for a F\"orster resonance using both one- and two-eigenstate models at Rabi frequency $\Omega_{\rm max}=2\pi\times10$ MHz, with results presented in Fig. \ref{fig:gates_1_vs_2}. Simulated gate fidelities differ by up to two orders of magnitude between the two atomic models. Imperfect F\"orster resonances, where $o_1\approx o_2$ yet $o_1\neq o_2$ still exhibit exchange dynamics and therefore provide increased gate fidelities over one-eigenstate models. An analysis of imperfect resonances is presented in \cite{Norrell2026SM}.

To explain the increase in gate fidelity in the two-eigenstate model, we consider the $\pi-2\pi-\pi$ gate in the $V>\Omega$ limit.
In the one-eigenstate model, the $2\pi$ pulse off-resonantly drives $\ket{a1}\mapsto\ket{ab}$, due to finite detuning $V$. This produces a residual linear AC-Stark phase on $\ket{a1}$ of order $\Omega/V$, the dominant error source compared to the population leakage of order $\left(\Omega/V\right)^2$.
In the two-eigenstate model, $\ket{ab}$ is not a stationary state; the F\"orster coupling $V$ continuously mixes $\ket{ab}$ and $\ket{\alpha \beta}$. There is therefore identical phase cancellation due to dark-state structure \cite{Petrosyan2017}, and the remaining error contributions are to state amplitudes of order $\left(\Omega/V\right)^2$. Full calculations of the presented scalings are given in  \cite{Norrell2026SM}. 

Optimization of TO gates under one- or two-eigenstate models yields identical gate performance in the blockade regime (Fig.~\ref{fig:gates_1_vs_2}(c)), which we attribute to the ability of the pulse to adapt to different models by changing optimization weights. 
Despite this, gates numerically optimized under the one-eigenstate model perform significantly worse when applied to the two-eigenstate model. This discrepancy emphasizes the importance of including all dominant interaction channels within the Rydberg manifold during gate design and optimization. Simplified one-eigenstate descriptions fail to accurately predict gate performance for Rydberg pair states with multiple interaction channels. Furthermore, due to the presence of local minima in the phase landscape, it is unlikely experimental optimizations of pulse waveforms will converge to maximum fidelity when starting from pulses optimized under an atomic model that does not match the dimensionality of the pair interaction space.

\rsub{We additionally analyze ARMD gates with smooth beam turn-on profiles, whose performance exhibits a similar increase when applied to two-eigenstate models. The form and performance of these gates is elaborated in detail in the Supplemental \cite{Norrell2026SM}.}

\textit{Conclusion} --
We have analyzed Rydberg gate design through the lens of two-eigenstate pair physics of F\"orster resonances. The main results are as follows. 

(i) The exchange interaction between manifolds within F\"orster resonances can strongly attenuate entangling gate errors. This mechanism is not predicted by models typically used in gate analyses. Atomic models that include this interaction are therefore required to provide accurate gate fidelity calculations.  

(ii) The two-eigenstate model presented here provides a fundamental gate limit of $\eta \geq \frac{\pi}{2}$ when allowing for rank-two access to the pair state Hilbert space. We provide exact pulse sequences that saturate this bound.

We note that for rank-one gate pulses, no additional hardware is required to realize enhanced gate fidelities other than addressing levels that are F\"orster resonant. Further, the increased gate fidelity is  most pronounced in the low $V$ regime, thereby improving long-range gates which is a key challenge in realizing error-correction codes for neutral atoms.

\textit{Acknowledgments --}
We acknowledge support from ARO under contract W911NF2410382, the National Science Foundation under Award 2016136 for the
QLCI center Hybrid Quantum Architectures and Networks, and Infleqtion.

The authors thank Trent Graham for useful discussions.
\bibliography{atomic, optics, qc_refs, rydberg, saffman_refs, this_paper}

%apsrev4-2.bst 2019-01-14 (MD) hand-edited version of apsrev4-1.bst
%Control: key (0)
%Control: author (8) initials jnrlst
%Control: editor formatted (1) identically to author
%Control: production of article title (0) allowed
%Control: page (0) single
%Control: year (1) truncated
%Control: production of eprint (0) enabled
\begin{thebibliography}{37}%
\makeatletter
\providecommand \@ifxundefined [1]{%
 \@ifx{#1\undefined}
}%
\providecommand \@ifnum [1]{%
 \ifnum #1\expandafter \@firstoftwo
 \else \expandafter \@secondoftwo
 \fi
}%
\providecommand \@ifx [1]{%
 \ifx #1\expandafter \@firstoftwo
 \else \expandafter \@secondoftwo
 \fi
}%
\providecommand \natexlab [1]{#1}%
\providecommand \enquote  [1]{``#1''}%
\providecommand \bibnamefont  [1]{#1}%
\providecommand \bibfnamefont [1]{#1}%
\providecommand \citenamefont [1]{#1}%
\providecommand \href@noop [0]{\@secondoftwo}%
\providecommand \href [0]{\begingroup \@sanitize@url \@href}%
\providecommand \@href[1]{\@@startlink{#1}\@@href}%
\providecommand \@@href[1]{\endgroup#1\@@endlink}%
\providecommand \@sanitize@url [0]{\catcode `\\12\catcode `\$12\catcode `\&12\catcode `\#12\catcode `\^12\catcode `\_12\catcode `\%12\relax}%
\providecommand \@@startlink[1]{}%
\providecommand \@@endlink[0]{}%
\providecommand \url  [0]{\begingroup\@sanitize@url \@url }%
\providecommand \@url [1]{\endgroup\@href {#1}{\urlprefix }}%
\providecommand \urlprefix  [0]{URL }%
\providecommand \Eprint [0]{\href }%
\providecommand \doibase [0]{https://doi.org/}%
\providecommand \selectlanguage [0]{\@gobble}%
\providecommand \bibinfo  [0]{\@secondoftwo}%
\providecommand \bibfield  [0]{\@secondoftwo}%
\providecommand \translation [1]{[#1]}%
\providecommand \BibitemOpen [0]{}%
\providecommand \bibitemStop [0]{}%
\providecommand \bibitemNoStop [0]{.\EOS\space}%
\providecommand \EOS [0]{\spacefactor3000\relax}%
\providecommand \BibitemShut  [1]{\csname bibitem#1\endcsname}%
\let\auto@bib@innerbib\@empty
%</preamble>
\bibitem [{\citenamefont {Evered}\ \emph {et~al.}(2023)\citenamefont {Evered}, \citenamefont {Bluvstein}, \citenamefont {Kalinowski}, \citenamefont {Ebadi}, \citenamefont {Manovitz}, \citenamefont {Zhou}, \citenamefont {Li}, \citenamefont {Geim}, \citenamefont {Wang}, \citenamefont {Maskara}, \citenamefont {Levine}, \citenamefont {Semeghini}, \citenamefont {Greiner}, \citenamefont {Vuleti\'c},\ and\ \citenamefont {Lukin}}]{Evered2023}%
  \BibitemOpen
  \bibfield  {author} {\bibinfo {author} {\bibfnamefont {S.~J.}\ \bibnamefont {Evered}}, \bibinfo {author} {\bibfnamefont {D.}~\bibnamefont {Bluvstein}}, \bibinfo {author} {\bibfnamefont {M.}~\bibnamefont {Kalinowski}}, \bibinfo {author} {\bibfnamefont {S.}~\bibnamefont {Ebadi}}, \bibinfo {author} {\bibfnamefont {T.}~\bibnamefont {Manovitz}}, \bibinfo {author} {\bibfnamefont {H.}~\bibnamefont {Zhou}}, \bibinfo {author} {\bibfnamefont {S.~H.}\ \bibnamefont {Li}}, \bibinfo {author} {\bibfnamefont {A.~A.}\ \bibnamefont {Geim}}, \bibinfo {author} {\bibfnamefont {T.~T.}\ \bibnamefont {Wang}}, \bibinfo {author} {\bibfnamefont {N.}~\bibnamefont {Maskara}}, \bibinfo {author} {\bibfnamefont {H.}~\bibnamefont {Levine}}, \bibinfo {author} {\bibfnamefont {G.}~\bibnamefont {Semeghini}}, \bibinfo {author} {\bibfnamefont {M.}~\bibnamefont {Greiner}}, \bibinfo {author} {\bibfnamefont {V.}~\bibnamefont {Vuleti\'c}},\ and\ \bibinfo {author} {\bibfnamefont {M.~D.}\ \bibnamefont {Lukin}},\ }\bibfield  {title} {\bibinfo {title}
  {High-fidelity parallel entangling gates on a neutral-atom quantum computer},\ }\href@noop {} {\bibfield  {journal} {\bibinfo  {journal} {Nature}\ }\textbf {\bibinfo {volume} {622}},\ \bibinfo {pages} {268} (\bibinfo {year} {2023})}\BibitemShut {NoStop}%
\bibitem [{\citenamefont {Peper}\ \emph {et~al.}(2025)\citenamefont {Peper}, \citenamefont {Li}, \citenamefont {Knapp}, \citenamefont {Bileska}, \citenamefont {Ma}, \citenamefont {Liu}, \citenamefont {Peng}, \citenamefont {Zhang}, \citenamefont {Horvath}, \citenamefont {Burgers},\ and\ \citenamefont {Thompson}}]{Peper2025}%
  \BibitemOpen
  \bibfield  {author} {\bibinfo {author} {\bibfnamefont {M.}~\bibnamefont {Peper}}, \bibinfo {author} {\bibfnamefont {Y.}~\bibnamefont {Li}}, \bibinfo {author} {\bibfnamefont {D.~Y.}\ \bibnamefont {Knapp}}, \bibinfo {author} {\bibfnamefont {M.}~\bibnamefont {Bileska}}, \bibinfo {author} {\bibfnamefont {S.}~\bibnamefont {Ma}}, \bibinfo {author} {\bibfnamefont {G.}~\bibnamefont {Liu}}, \bibinfo {author} {\bibfnamefont {P.}~\bibnamefont {Peng}}, \bibinfo {author} {\bibfnamefont {B.}~\bibnamefont {Zhang}}, \bibinfo {author} {\bibfnamefont {S.~P.}\ \bibnamefont {Horvath}}, \bibinfo {author} {\bibfnamefont {A.~P.}\ \bibnamefont {Burgers}},\ and\ \bibinfo {author} {\bibfnamefont {J.~D.}\ \bibnamefont {Thompson}},\ }\bibfield  {title} {\bibinfo {title} {Spectroscopy and modeling of $^{171}\mathrm{Yb}$ {R}ydberg states for high-fidelity two-qubit gates},\ }\href@noop {} {\bibfield  {journal} {\bibinfo  {journal} {Phys. Rev. X}\ }\textbf {\bibinfo {volume} {15}},\ \bibinfo {pages} {011009} (\bibinfo {year}
  {2025})}\BibitemShut {NoStop}%
\bibitem [{\citenamefont {Muniz}\ \emph {et~al.}(2025)\citenamefont {Muniz}, \citenamefont {Stone}, \citenamefont {Stack}, \citenamefont {Jaffe}, \citenamefont {Kindem}, \citenamefont {Wadleigh}, \citenamefont {Zalys-Geller}, \citenamefont {Zhang}, \citenamefont {Chen}, \citenamefont {Norcia}, \citenamefont {Epstein}, \citenamefont {Halperin}, \citenamefont {Hummel}, \citenamefont {Wilkason}, \citenamefont {Li}, \citenamefont {Barnes}, \citenamefont {Battaglino}, \citenamefont {Bohdanowicz}, \citenamefont {Booth}, \citenamefont {Brown}, \citenamefont {Brown}, \citenamefont {Cairncross}, \citenamefont {Cassella}, \citenamefont {Coxe}, \citenamefont {Crow}, \citenamefont {Feldkamp}, \citenamefont {Griger}, \citenamefont {Heinz}, \citenamefont {Jones}, \citenamefont {Kim}, \citenamefont {King}, \citenamefont {Kotru}, \citenamefont {Lauigan}, \citenamefont {Marjanovic}, \citenamefont {Megidish}, \citenamefont {Meredith}, \citenamefont {McDonald}, \citenamefont {Morshead}, \citenamefont {Narayanaswami},
  \citenamefont {Nishiguchi}, \citenamefont {Paule}, \citenamefont {Pawlak}, \citenamefont {Pudenz}, \citenamefont {P\'erez}, \citenamefont {Ryou}, \citenamefont {Simon}, \citenamefont {Smull}, \citenamefont {Urbanek}, \citenamefont {van~de Veerdonk}, \citenamefont {Vendeiro}, \citenamefont {Wu}, \citenamefont {Xie},\ and\ \citenamefont {Bloom}}]{Muniz2025}%
  \BibitemOpen
  \bibfield  {author} {\bibinfo {author} {\bibfnamefont {J.~A.}\ \bibnamefont {Muniz}}, \bibinfo {author} {\bibfnamefont {M.}~\bibnamefont {Stone}}, \bibinfo {author} {\bibfnamefont {D.~T.}\ \bibnamefont {Stack}}, \bibinfo {author} {\bibfnamefont {M.}~\bibnamefont {Jaffe}}, \bibinfo {author} {\bibfnamefont {J.~M.}\ \bibnamefont {Kindem}}, \bibinfo {author} {\bibfnamefont {L.}~\bibnamefont {Wadleigh}}, \bibinfo {author} {\bibfnamefont {E.}~\bibnamefont {Zalys-Geller}}, \bibinfo {author} {\bibfnamefont {X.}~\bibnamefont {Zhang}}, \bibinfo {author} {\bibfnamefont {C.-A.}\ \bibnamefont {Chen}}, \bibinfo {author} {\bibfnamefont {M.~A.}\ \bibnamefont {Norcia}}, \bibinfo {author} {\bibfnamefont {J.}~\bibnamefont {Epstein}}, \bibinfo {author} {\bibfnamefont {E.}~\bibnamefont {Halperin}}, \bibinfo {author} {\bibfnamefont {F.}~\bibnamefont {Hummel}}, \bibinfo {author} {\bibfnamefont {T.}~\bibnamefont {Wilkason}}, \bibinfo {author} {\bibfnamefont {M.}~\bibnamefont {Li}}, \bibinfo {author} {\bibfnamefont
  {K.}~\bibnamefont {Barnes}}, \bibinfo {author} {\bibfnamefont {P.}~\bibnamefont {Battaglino}}, \bibinfo {author} {\bibfnamefont {T.~C.}\ \bibnamefont {Bohdanowicz}}, \bibinfo {author} {\bibfnamefont {G.}~\bibnamefont {Booth}}, \bibinfo {author} {\bibfnamefont {A.}~\bibnamefont {Brown}}, \bibinfo {author} {\bibfnamefont {M.~O.}\ \bibnamefont {Brown}}, \bibinfo {author} {\bibfnamefont {W.~B.}\ \bibnamefont {Cairncross}}, \bibinfo {author} {\bibfnamefont {K.}~\bibnamefont {Cassella}}, \bibinfo {author} {\bibfnamefont {R.}~\bibnamefont {Coxe}}, \bibinfo {author} {\bibfnamefont {D.}~\bibnamefont {Crow}}, \bibinfo {author} {\bibfnamefont {M.}~\bibnamefont {Feldkamp}}, \bibinfo {author} {\bibfnamefont {C.}~\bibnamefont {Griger}}, \bibinfo {author} {\bibfnamefont {A.}~\bibnamefont {Heinz}}, \bibinfo {author} {\bibfnamefont {A.~M.~W.}\ \bibnamefont {Jones}}, \bibinfo {author} {\bibfnamefont {H.}~\bibnamefont {Kim}}, \bibinfo {author} {\bibfnamefont {J.}~\bibnamefont {King}}, \bibinfo {author} {\bibfnamefont
  {K.}~\bibnamefont {Kotru}}, \bibinfo {author} {\bibfnamefont {J.}~\bibnamefont {Lauigan}}, \bibinfo {author} {\bibfnamefont {J.}~\bibnamefont {Marjanovic}}, \bibinfo {author} {\bibfnamefont {E.}~\bibnamefont {Megidish}}, \bibinfo {author} {\bibfnamefont {M.}~\bibnamefont {Meredith}}, \bibinfo {author} {\bibfnamefont {M.}~\bibnamefont {McDonald}}, \bibinfo {author} {\bibfnamefont {R.}~\bibnamefont {Morshead}}, \bibinfo {author} {\bibfnamefont {S.}~\bibnamefont {Narayanaswami}}, \bibinfo {author} {\bibfnamefont {C.}~\bibnamefont {Nishiguchi}}, \bibinfo {author} {\bibfnamefont {T.}~\bibnamefont {Paule}}, \bibinfo {author} {\bibfnamefont {K.~A.}\ \bibnamefont {Pawlak}}, \bibinfo {author} {\bibfnamefont {K.~L.}\ \bibnamefont {Pudenz}}, \bibinfo {author} {\bibfnamefont {D.~R.}\ \bibnamefont {P\'erez}}, \bibinfo {author} {\bibfnamefont {A.}~\bibnamefont {Ryou}}, \bibinfo {author} {\bibfnamefont {J.}~\bibnamefont {Simon}}, \bibinfo {author} {\bibfnamefont {A.}~\bibnamefont {Smull}}, \bibinfo {author} {\bibfnamefont
  {M.}~\bibnamefont {Urbanek}}, \bibinfo {author} {\bibfnamefont {R.~J.~M.}\ \bibnamefont {van~de Veerdonk}}, \bibinfo {author} {\bibfnamefont {Z.}~\bibnamefont {Vendeiro}}, \bibinfo {author} {\bibfnamefont {T.-Y.}\ \bibnamefont {Wu}}, \bibinfo {author} {\bibfnamefont {X.}~\bibnamefont {Xie}},\ and\ \bibinfo {author} {\bibfnamefont {B.~J.}\ \bibnamefont {Bloom}},\ }\bibfield  {title} {\bibinfo {title} {High-fidelity universal gates in the ${}^{171}$$\mathrm{Yb}$ ground-state nuclear-spin qubit},\ }\href@noop {} {\bibfield  {journal} {\bibinfo  {journal} {PRX Quantum}\ }\textbf {\bibinfo {volume} {6}},\ \bibinfo {pages} {020334} (\bibinfo {year} {2025})}\BibitemShut {NoStop}%
\bibitem [{\citenamefont {Tsai}\ \emph {et~al.}(2025)\citenamefont {Tsai}, \citenamefont {Sun}, \citenamefont {Shaw}, \citenamefont {Finkelstein},\ and\ \citenamefont {Endres}}]{Tsai2025}%
  \BibitemOpen
  \bibfield  {author} {\bibinfo {author} {\bibfnamefont {R.~B.-S.}\ \bibnamefont {Tsai}}, \bibinfo {author} {\bibfnamefont {X.}~\bibnamefont {Sun}}, \bibinfo {author} {\bibfnamefont {A.~L.}\ \bibnamefont {Shaw}}, \bibinfo {author} {\bibfnamefont {R.}~\bibnamefont {Finkelstein}},\ and\ \bibinfo {author} {\bibfnamefont {M.}~\bibnamefont {Endres}},\ }\bibfield  {title} {\bibinfo {title} {Benchmarking and fidelity response theory of high-fidelity {R}ydberg entangling gates},\ }\href@noop {} {\bibfield  {journal} {\bibinfo  {journal} {PRX Quantum}\ }\textbf {\bibinfo {volume} {6}},\ \bibinfo {pages} {010331} (\bibinfo {year} {2025})}\BibitemShut {NoStop}%
\bibitem [{\citenamefont {Radnaev}\ \emph {et~al.}(2025)\citenamefont {Radnaev}, \citenamefont {Chung}, \citenamefont {Cole}, \citenamefont {Mason}, \citenamefont {Ballance}, \citenamefont {Bedalov}, \citenamefont {Belknap}, \citenamefont {Berman}, \citenamefont {Blakely}, \citenamefont {Bloomfield}, \citenamefont {Buttler}, \citenamefont {Campbell}, \citenamefont {Chopinaud}, \citenamefont {Copenhaver}, \citenamefont {Dawes}, \citenamefont {Eubanks}, \citenamefont {Friss}, \citenamefont {Garcia}, \citenamefont {Gilbert}, \citenamefont {Gillette}, \citenamefont {Goiporia}, \citenamefont {Gokhale}, \citenamefont {Goldwin}, \citenamefont {Goodwin}, \citenamefont {Graham}, \citenamefont {Guttormsson}, \citenamefont {Hickman}, \citenamefont {Hurtley}, \citenamefont {Iliev}, \citenamefont {Jones}, \citenamefont {Jones}, \citenamefont {Kuper}, \citenamefont {Lewis}, \citenamefont {Lichtman}, \citenamefont {Majdeteimouri}, \citenamefont {Mason}, \citenamefont {McMaster}, \citenamefont {Miles}, \citenamefont
  {Mitchell}, \citenamefont {Murphree}, \citenamefont {Neff-Mallon}, \citenamefont {Oh}, \citenamefont {Omole}, \citenamefont {Simon}, \citenamefont {Pederson}, \citenamefont {Perlin}, \citenamefont {Reiter}, \citenamefont {Rines}, \citenamefont {Romlow}, \citenamefont {Scott}, \citenamefont {Stiefvater}, \citenamefont {Tanner}, \citenamefont {Tucker}, \citenamefont {Vinogradov}, \citenamefont {Warter}, \citenamefont {Yeo}, \citenamefont {Saffman},\ and\ \citenamefont {Noel}}]{Radnaev2025}%
  \BibitemOpen
  \bibfield  {author} {\bibinfo {author} {\bibfnamefont {A.~G.}\ \bibnamefont {Radnaev}}, \bibinfo {author} {\bibfnamefont {W.~C.}\ \bibnamefont {Chung}}, \bibinfo {author} {\bibfnamefont {D.~C.}\ \bibnamefont {Cole}}, \bibinfo {author} {\bibfnamefont {D.}~\bibnamefont {Mason}}, \bibinfo {author} {\bibfnamefont {T.~G.}\ \bibnamefont {Ballance}}, \bibinfo {author} {\bibfnamefont {M.~J.}\ \bibnamefont {Bedalov}}, \bibinfo {author} {\bibfnamefont {D.~A.}\ \bibnamefont {Belknap}}, \bibinfo {author} {\bibfnamefont {M.~R.}\ \bibnamefont {Berman}}, \bibinfo {author} {\bibfnamefont {M.}~\bibnamefont {Blakely}}, \bibinfo {author} {\bibfnamefont {I.~L.}\ \bibnamefont {Bloomfield}}, \bibinfo {author} {\bibfnamefont {P.~D.}\ \bibnamefont {Buttler}}, \bibinfo {author} {\bibfnamefont {C.}~\bibnamefont {Campbell}}, \bibinfo {author} {\bibfnamefont {A.}~\bibnamefont {Chopinaud}}, \bibinfo {author} {\bibfnamefont {E.}~\bibnamefont {Copenhaver}}, \bibinfo {author} {\bibfnamefont {M.~K.}\ \bibnamefont {Dawes}}, \bibinfo
  {author} {\bibfnamefont {S.~Y.}\ \bibnamefont {Eubanks}}, \bibinfo {author} {\bibfnamefont {A.~J.}\ \bibnamefont {Friss}}, \bibinfo {author} {\bibfnamefont {D.~M.}\ \bibnamefont {Garcia}}, \bibinfo {author} {\bibfnamefont {J.}~\bibnamefont {Gilbert}}, \bibinfo {author} {\bibfnamefont {M.}~\bibnamefont {Gillette}}, \bibinfo {author} {\bibfnamefont {P.}~\bibnamefont {Goiporia}}, \bibinfo {author} {\bibfnamefont {P.}~\bibnamefont {Gokhale}}, \bibinfo {author} {\bibfnamefont {J.}~\bibnamefont {Goldwin}}, \bibinfo {author} {\bibfnamefont {D.}~\bibnamefont {Goodwin}}, \bibinfo {author} {\bibfnamefont {T.~M.}\ \bibnamefont {Graham}}, \bibinfo {author} {\bibfnamefont {C.}~\bibnamefont {Guttormsson}}, \bibinfo {author} {\bibfnamefont {G.~T.}\ \bibnamefont {Hickman}}, \bibinfo {author} {\bibfnamefont {L.}~\bibnamefont {Hurtley}}, \bibinfo {author} {\bibfnamefont {M.}~\bibnamefont {Iliev}}, \bibinfo {author} {\bibfnamefont {E.~B.}\ \bibnamefont {Jones}}, \bibinfo {author} {\bibfnamefont {R.~A.}\ \bibnamefont {Jones}},
  \bibinfo {author} {\bibfnamefont {K.~W.}\ \bibnamefont {Kuper}}, \bibinfo {author} {\bibfnamefont {T.~B.}\ \bibnamefont {Lewis}}, \bibinfo {author} {\bibfnamefont {M.~T.}\ \bibnamefont {Lichtman}}, \bibinfo {author} {\bibfnamefont {F.}~\bibnamefont {Majdeteimouri}}, \bibinfo {author} {\bibfnamefont {J.~J.}\ \bibnamefont {Mason}}, \bibinfo {author} {\bibfnamefont {J.~K.}\ \bibnamefont {McMaster}}, \bibinfo {author} {\bibfnamefont {J.~A.}\ \bibnamefont {Miles}}, \bibinfo {author} {\bibfnamefont {P.~T.}\ \bibnamefont {Mitchell}}, \bibinfo {author} {\bibfnamefont {J.~D.}\ \bibnamefont {Murphree}}, \bibinfo {author} {\bibfnamefont {N.~A.}\ \bibnamefont {Neff-Mallon}}, \bibinfo {author} {\bibfnamefont {T.}~\bibnamefont {Oh}}, \bibinfo {author} {\bibfnamefont {V.}~\bibnamefont {Omole}}, \bibinfo {author} {\bibfnamefont {C.~P.}\ \bibnamefont {Simon}}, \bibinfo {author} {\bibfnamefont {N.}~\bibnamefont {Pederson}}, \bibinfo {author} {\bibfnamefont {M.~A.}\ \bibnamefont {Perlin}}, \bibinfo {author} {\bibfnamefont
  {A.}~\bibnamefont {Reiter}}, \bibinfo {author} {\bibfnamefont {R.}~\bibnamefont {Rines}}, \bibinfo {author} {\bibfnamefont {P.}~\bibnamefont {Romlow}}, \bibinfo {author} {\bibfnamefont {A.~M.}\ \bibnamefont {Scott}}, \bibinfo {author} {\bibfnamefont {D.}~\bibnamefont {Stiefvater}}, \bibinfo {author} {\bibfnamefont {J.~R.}\ \bibnamefont {Tanner}}, \bibinfo {author} {\bibfnamefont {A.~K.}\ \bibnamefont {Tucker}}, \bibinfo {author} {\bibfnamefont {I.~V.}\ \bibnamefont {Vinogradov}}, \bibinfo {author} {\bibfnamefont {M.~L.}\ \bibnamefont {Warter}}, \bibinfo {author} {\bibfnamefont {M.}~\bibnamefont {Yeo}}, \bibinfo {author} {\bibfnamefont {M.}~\bibnamefont {Saffman}},\ and\ \bibinfo {author} {\bibfnamefont {T.~W.}\ \bibnamefont {Noel}},\ }\bibfield  {title} {\bibinfo {title} {Universal neutral-atom quantum computer with individual optical addressing and nondestructive readout},\ }\href@noop {} {\bibfield  {journal} {\bibinfo  {journal} {PRX Quantum}\ }\textbf {\bibinfo {volume} {6}},\ \bibinfo {pages} {030334}
  (\bibinfo {year} {2025})}\BibitemShut {NoStop}%
\bibitem [{\citenamefont {Saffman}\ and\ \citenamefont {Walker}(2005)}]{saffman2005a}%
  \BibitemOpen
  \bibfield  {author} {\bibinfo {author} {\bibfnamefont {M.}~\bibnamefont {Saffman}}\ and\ \bibinfo {author} {\bibfnamefont {T.~G.}\ \bibnamefont {Walker}},\ }\bibfield  {title} {\bibinfo {title} {Analysis of a quantum logic device based on dipole-dipole interactions of optically trapped {R}ydberg atoms},\ }\href@noop {} {\bibfield  {journal} {\bibinfo  {journal} {Phys. Rev. A}\ }\textbf {\bibinfo {volume} {72}},\ \bibinfo {pages} {022347} (\bibinfo {year} {2005})}\BibitemShut {NoStop}%
\bibitem [{\citenamefont {Doultsinos}\ and\ \citenamefont {Petrosyan}(2025)}]{Doultsinos2025}%
  \BibitemOpen
  \bibfield  {author} {\bibinfo {author} {\bibfnamefont {G.}~\bibnamefont {Doultsinos}}\ and\ \bibinfo {author} {\bibfnamefont {D.}~\bibnamefont {Petrosyan}},\ }\bibfield  {title} {\bibinfo {title} {Quantum gates between distant atoms mediated by a {R}ydberg excitation antiferromagnet},\ }\href@noop {} {\bibfield  {journal} {\bibinfo  {journal} {Phys. Rev. Res.}\ }\textbf {\bibinfo {volume} {7}},\ \bibinfo {pages} {023246} (\bibinfo {year} {2025})}\BibitemShut {NoStop}%
\bibitem [{\citenamefont {Jandura}\ and\ \citenamefont {Pupillo}(2022)}]{Jandura2022}%
  \BibitemOpen
  \bibfield  {author} {\bibinfo {author} {\bibfnamefont {S.}~\bibnamefont {Jandura}}\ and\ \bibinfo {author} {\bibfnamefont {G.}~\bibnamefont {Pupillo}},\ }\bibfield  {title} {\bibinfo {title} {Time-optimal two- and three-qubit gates for {R}ydberg atoms},\ }\href@noop {} {\bibfield  {journal} {\bibinfo  {journal} {Quantum}\ }\textbf {\bibinfo {volume} {6}},\ \bibinfo {pages} {712} (\bibinfo {year} {2022})}\BibitemShut {NoStop}%
\bibitem [{\citenamefont {Jandura}\ \emph {et~al.}(2023)\citenamefont {Jandura}, \citenamefont {Thompson},\ and\ \citenamefont {Pupillo}}]{Jandura2023}%
  \BibitemOpen
  \bibfield  {author} {\bibinfo {author} {\bibfnamefont {S.}~\bibnamefont {Jandura}}, \bibinfo {author} {\bibfnamefont {J.~D.}\ \bibnamefont {Thompson}},\ and\ \bibinfo {author} {\bibfnamefont {G.}~\bibnamefont {Pupillo}},\ }\bibfield  {title} {\bibinfo {title} {Optimizing {R}ydberg gates for logical-qubit performance},\ }\href@noop {} {\bibfield  {journal} {\bibinfo  {journal} {PRX Quantum}\ }\textbf {\bibinfo {volume} {4}},\ \bibinfo {pages} {020336} (\bibinfo {year} {2023})}\BibitemShut {NoStop}%
\bibitem [{\citenamefont {Pagano}\ \emph {et~al.}(2022)\citenamefont {Pagano}, \citenamefont {Weber}, \citenamefont {Jaschke}, \citenamefont {Pfau}, \citenamefont {Meinert}, \citenamefont {Montangero},\ and\ \citenamefont {B\"uchler}}]{Pagano2022}%
  \BibitemOpen
  \bibfield  {author} {\bibinfo {author} {\bibfnamefont {A.}~\bibnamefont {Pagano}}, \bibinfo {author} {\bibfnamefont {S.}~\bibnamefont {Weber}}, \bibinfo {author} {\bibfnamefont {D.}~\bibnamefont {Jaschke}}, \bibinfo {author} {\bibfnamefont {T.}~\bibnamefont {Pfau}}, \bibinfo {author} {\bibfnamefont {F.}~\bibnamefont {Meinert}}, \bibinfo {author} {\bibfnamefont {S.}~\bibnamefont {Montangero}},\ and\ \bibinfo {author} {\bibfnamefont {H.~P.}\ \bibnamefont {B\"uchler}},\ }\bibfield  {title} {\bibinfo {title} {Error budgeting for a controlled-phase gate with {S}trontium-88 {R}ydberg atoms},\ }\href@noop {} {\bibfield  {journal} {\bibinfo  {journal} {Phys. Rev. Res.}\ }\textbf {\bibinfo {volume} {4}},\ \bibinfo {pages} {033019} (\bibinfo {year} {2022})}\BibitemShut {NoStop}%
\bibitem [{\citenamefont {Mohan}\ \emph {et~al.}(2023)\citenamefont {Mohan}, \citenamefont {de~Keijzer},\ and\ \citenamefont {Kokkelmans}}]{Mohan2023}%
  \BibitemOpen
  \bibfield  {author} {\bibinfo {author} {\bibfnamefont {M.}~\bibnamefont {Mohan}}, \bibinfo {author} {\bibfnamefont {R.}~\bibnamefont {de~Keijzer}},\ and\ \bibinfo {author} {\bibfnamefont {S.}~\bibnamefont {Kokkelmans}},\ }\bibfield  {title} {\bibinfo {title} {Robust control and optimal {R}ydberg states for neutral atom two-qubit gates},\ }\href@noop {} {\bibfield  {journal} {\bibinfo  {journal} {Phys. Rev. Res.}\ }\textbf {\bibinfo {volume} {5}},\ \bibinfo {pages} {033052} (\bibinfo {year} {2023})}\BibitemShut {NoStop}%
\bibitem [{\citenamefont {Cole}\ \emph {et~al.}(2026)\citenamefont {Cole}, \citenamefont {Buchemmavari},\ and\ \citenamefont {Saffman}}]{Cole2026}%
  \BibitemOpen
  \bibfield  {author} {\bibinfo {author} {\bibfnamefont {D.~C.}\ \bibnamefont {Cole}}, \bibinfo {author} {\bibfnamefont {V.}~\bibnamefont {Buchemmavari}},\ and\ \bibinfo {author} {\bibfnamefont {M.}~\bibnamefont {Saffman}},\ }\bibfield  {title} {\bibinfo {title} {An asymmetric and fast {R}ydberg gate protocol for entanglement outside of the blockade regime},\ }\href@noop {} {\bibfield  {journal} {\bibinfo  {journal} {arXiv:2512.22767}\ } (\bibinfo {year} {2026})}\BibitemShut {NoStop}%
\bibitem [{\citenamefont {Petrosyan}\ \emph {et~al.}(2017)\citenamefont {Petrosyan}, \citenamefont {Motzoi}, \citenamefont {Saffman},\ and\ \citenamefont {M\o{}lmer}}]{Petrosyan2017}%
  \BibitemOpen
  \bibfield  {author} {\bibinfo {author} {\bibfnamefont {D.}~\bibnamefont {Petrosyan}}, \bibinfo {author} {\bibfnamefont {F.}~\bibnamefont {Motzoi}}, \bibinfo {author} {\bibfnamefont {M.}~\bibnamefont {Saffman}},\ and\ \bibinfo {author} {\bibfnamefont {K.}~\bibnamefont {M\o{}lmer}},\ }\bibfield  {title} {\bibinfo {title} {High-fidelity {R}ydberg quantum gate via a two-atom dark state},\ }\href@noop {} {\bibfield  {journal} {\bibinfo  {journal} {Phys. Rev. A}\ }\textbf {\bibinfo {volume} {96}},\ \bibinfo {pages} {042306} (\bibinfo {year} {2017})}\BibitemShut {NoStop}%
\bibitem [{\citenamefont {Walker}\ and\ \citenamefont {Saffman}(2008)}]{Walker2008}%
  \BibitemOpen
  \bibfield  {author} {\bibinfo {author} {\bibfnamefont {T.~G.}\ \bibnamefont {Walker}}\ and\ \bibinfo {author} {\bibfnamefont {M.}~\bibnamefont {Saffman}},\ }\bibfield  {title} {\bibinfo {title} {Consequences of {Z}eeman degeneracy for the van der {W}aals blockade between {R}ydberg atoms},\ }\href@noop {} {\bibfield  {journal} {\bibinfo  {journal} {Phys. Rev. A}\ }\textbf {\bibinfo {volume} {77}},\ \bibinfo {pages} {032723} (\bibinfo {year} {2008})}\BibitemShut {NoStop}%
\bibitem [{\citenamefont {Jaksch}\ \emph {et~al.}(2000)\citenamefont {Jaksch}, \citenamefont {Cirac}, \citenamefont {Zoller}, \citenamefont {Rolston}, \citenamefont {C\^ot\'e},\ and\ \citenamefont {Lukin}}]{Jaksch2000}%
  \BibitemOpen
  \bibfield  {author} {\bibinfo {author} {\bibfnamefont {D.}~\bibnamefont {Jaksch}}, \bibinfo {author} {\bibfnamefont {J.~I.}\ \bibnamefont {Cirac}}, \bibinfo {author} {\bibfnamefont {P.}~\bibnamefont {Zoller}}, \bibinfo {author} {\bibfnamefont {S.~L.}\ \bibnamefont {Rolston}}, \bibinfo {author} {\bibfnamefont {R.}~\bibnamefont {C\^ot\'e}},\ and\ \bibinfo {author} {\bibfnamefont {M.~D.}\ \bibnamefont {Lukin}},\ }\bibfield  {title} {\bibinfo {title} {Fast quantum gates for neutral atoms},\ }\href@noop {} {\bibfield  {journal} {\bibinfo  {journal} {Phys. Rev. Lett.}\ }\textbf {\bibinfo {volume} {85}},\ \bibinfo {pages} {2208} (\bibinfo {year} {2000})}\BibitemShut {NoStop}%
\bibitem [{\citenamefont {Shi}(2018)}]{XFShi2018}%
  \BibitemOpen
  \bibfield  {author} {\bibinfo {author} {\bibfnamefont {X.-F.}\ \bibnamefont {Shi}},\ }\bibfield  {title} {\bibinfo {title} {Deutsch, {T}offoli, and {CNOT} gates via {R}ydberg blockade of neutral atoms},\ }\href@noop {} {\bibfield  {journal} {\bibinfo  {journal} {Phys. Rev. Appl.}\ }\textbf {\bibinfo {volume} {9}},\ \bibinfo {pages} {051001} (\bibinfo {year} {2018})}\BibitemShut {NoStop}%
\bibitem [{\citenamefont {Ravets}\ \emph {et~al.}(2014)\citenamefont {Ravets}, \citenamefont {Labuhn}, \citenamefont {Barredo}, \citenamefont {B\'eguin}, \citenamefont {Lahaye},\ and\ \citenamefont {Browaeys}}]{Ravets2014}%
  \BibitemOpen
  \bibfield  {author} {\bibinfo {author} {\bibfnamefont {S.}~\bibnamefont {Ravets}}, \bibinfo {author} {\bibfnamefont {H.}~\bibnamefont {Labuhn}}, \bibinfo {author} {\bibfnamefont {D.}~\bibnamefont {Barredo}}, \bibinfo {author} {\bibfnamefont {L.}~\bibnamefont {B\'eguin}}, \bibinfo {author} {\bibfnamefont {T.}~\bibnamefont {Lahaye}},\ and\ \bibinfo {author} {\bibfnamefont {A.}~\bibnamefont {Browaeys}},\ }\bibfield  {title} {\bibinfo {title} {Coherent dipole dipole coupling between two single {R}ydberg atoms at an electrically-tuned {F}\"orster resonance},\ }\href@noop {} {\bibfield  {journal} {\bibinfo  {journal} {Nat. Phys.}\ }\textbf {\bibinfo {volume} {10}},\ \bibinfo {pages} {914} (\bibinfo {year} {2014})}\BibitemShut {NoStop}%
\bibitem [{\citenamefont {Ryabtsev}\ \emph {et~al.}(2010)\citenamefont {Ryabtsev}, \citenamefont {Tretyakov}, \citenamefont {Beterov},\ and\ \citenamefont {Entin}}]{Ryabtsev2010}%
  \BibitemOpen
  \bibfield  {author} {\bibinfo {author} {\bibfnamefont {I.~I.}\ \bibnamefont {Ryabtsev}}, \bibinfo {author} {\bibfnamefont {D.~B.}\ \bibnamefont {Tretyakov}}, \bibinfo {author} {\bibfnamefont {I.~I.}\ \bibnamefont {Beterov}},\ and\ \bibinfo {author} {\bibfnamefont {V.~M.}\ \bibnamefont {Entin}},\ }\bibfield  {title} {\bibinfo {title} {Observation of the {S}tark-tuned {F}\"orster resonance between two {R}ydberg atoms},\ }\href@noop {} {\bibfield  {journal} {\bibinfo  {journal} {Phys. Rev. Lett.}\ }\textbf {\bibinfo {volume} {104}},\ \bibinfo {pages} {073003} (\bibinfo {year} {2010})}\BibitemShut {NoStop}%
\bibitem [{\citenamefont {Nipper}\ \emph {et~al.}(2012)\citenamefont {Nipper}, \citenamefont {Balewski}, \citenamefont {Krupp}, \citenamefont {Butscher}, \citenamefont {L\"ow},\ and\ \citenamefont {Pfau}}]{Nipper2012}%
  \BibitemOpen
  \bibfield  {author} {\bibinfo {author} {\bibfnamefont {J.}~\bibnamefont {Nipper}}, \bibinfo {author} {\bibfnamefont {J.~B.}\ \bibnamefont {Balewski}}, \bibinfo {author} {\bibfnamefont {A.~T.}\ \bibnamefont {Krupp}}, \bibinfo {author} {\bibfnamefont {B.}~\bibnamefont {Butscher}}, \bibinfo {author} {\bibfnamefont {R.}~\bibnamefont {L\"ow}},\ and\ \bibinfo {author} {\bibfnamefont {T.}~\bibnamefont {Pfau}},\ }\bibfield  {title} {\bibinfo {title} {Highly resolved measurements of {S}tark-tuned {F}\"orster resonances between {R}ydberg atoms},\ }\href {https://doi.org/10.1103/PhysRevLett.108.113001} {\bibfield  {journal} {\bibinfo  {journal} {Phys. Rev. Lett.}\ }\textbf {\bibinfo {volume} {108}},\ \bibinfo {pages} {113001} (\bibinfo {year} {2012})}\BibitemShut {NoStop}%
\bibitem [{\citenamefont {Ravets}\ \emph {et~al.}(2015)\citenamefont {Ravets}, \citenamefont {Labuhn}, \citenamefont {Barredo}, \citenamefont {Lahaye},\ and\ \citenamefont {Browaeys}}]{Ravets2015}%
  \BibitemOpen
  \bibfield  {author} {\bibinfo {author} {\bibfnamefont {S.}~\bibnamefont {Ravets}}, \bibinfo {author} {\bibfnamefont {H.}~\bibnamefont {Labuhn}}, \bibinfo {author} {\bibfnamefont {D.}~\bibnamefont {Barredo}}, \bibinfo {author} {\bibfnamefont {T.}~\bibnamefont {Lahaye}},\ and\ \bibinfo {author} {\bibfnamefont {A.}~\bibnamefont {Browaeys}},\ }\bibfield  {title} {\bibinfo {title} {Measurement of the angular dependence of the dipole-dipole interaction between two individual {R}ydberg atoms at a {F}\"orster resonance},\ }\href {https://doi.org/10.1103/PhysRevA.92.020701} {\bibfield  {journal} {\bibinfo  {journal} {Phys. Rev. A}\ }\textbf {\bibinfo {volume} {92}},\ \bibinfo {pages} {020701(R)} (\bibinfo {year} {2015})}\BibitemShut {NoStop}%
\bibitem [{\citenamefont {Beterov}\ and\ \citenamefont {Saffman}(2015)}]{Beterov2015}%
  \BibitemOpen
  \bibfield  {author} {\bibinfo {author} {\bibfnamefont {I.~I.}\ \bibnamefont {Beterov}}\ and\ \bibinfo {author} {\bibfnamefont {M.}~\bibnamefont {Saffman}},\ }\bibfield  {title} {\bibinfo {title} {{R}ydberg blockade, {F}\"orster resonances, and quantum state measurements with different atomic species},\ }\href@noop {} {\bibfield  {journal} {\bibinfo  {journal} {Phys. Rev. A}\ }\textbf {\bibinfo {volume} {92}},\ \bibinfo {pages} {042710} (\bibinfo {year} {2015})}\BibitemShut {NoStop}%
\bibitem [{\citenamefont {Beterov}\ \emph {et~al.}(2016)\citenamefont {Beterov}, \citenamefont {Saffman}, \citenamefont {Yakshina}, \citenamefont {Tretyakov}, \citenamefont {Entin}, \citenamefont {Bergamini}, \citenamefont {Kuznetsova},\ and\ \citenamefont {Ryabtsev}}]{Beterov2016b}%
  \BibitemOpen
  \bibfield  {author} {\bibinfo {author} {\bibfnamefont {I.~I.}\ \bibnamefont {Beterov}}, \bibinfo {author} {\bibfnamefont {M.}~\bibnamefont {Saffman}}, \bibinfo {author} {\bibfnamefont {E.~A.}\ \bibnamefont {Yakshina}}, \bibinfo {author} {\bibfnamefont {D.~B.}\ \bibnamefont {Tretyakov}}, \bibinfo {author} {\bibfnamefont {V.~M.}\ \bibnamefont {Entin}}, \bibinfo {author} {\bibfnamefont {S.}~\bibnamefont {Bergamini}}, \bibinfo {author} {\bibfnamefont {E.~A.}\ \bibnamefont {Kuznetsova}},\ and\ \bibinfo {author} {\bibfnamefont {I.~I.}\ \bibnamefont {Ryabtsev}},\ }\bibfield  {title} {\bibinfo {title} {Two-qubit gates using adiabatic passage of the {S}tark-tuned {F}\"orster resonances in {R}ydberg atoms},\ }\href@noop {} {\bibfield  {journal} {\bibinfo  {journal} {Phys. Rev. A}\ }\textbf {\bibinfo {volume} {94}},\ \bibinfo {pages} {062307} (\bibinfo {year} {2016})}\BibitemShut {NoStop}%
\bibitem [{\citenamefont {Browaeys}\ \emph {et~al.}(2016)\citenamefont {Browaeys}, \citenamefont {Barredo},\ and\ \citenamefont {Lahaye}}]{Browaeys2016}%
  \BibitemOpen
  \bibfield  {author} {\bibinfo {author} {\bibfnamefont {A.}~\bibnamefont {Browaeys}}, \bibinfo {author} {\bibfnamefont {D.}~\bibnamefont {Barredo}},\ and\ \bibinfo {author} {\bibfnamefont {T.}~\bibnamefont {Lahaye}},\ }\bibfield  {title} {\bibinfo {title} {Experimental investigations of dipole-dipole interactions between a few {R}ydberg atoms},\ }\href@noop {} {\bibfield  {journal} {\bibinfo  {journal} {J. Phys. B: At. Mol. Opt. Phys.}\ }\textbf {\bibinfo {volume} {49}},\ \bibinfo {pages} {152001} (\bibinfo {year} {2016})}\BibitemShut {NoStop}%
\bibitem [{\citenamefont {Anand}\ \emph {et~al.}(2024)\citenamefont {Anand}, \citenamefont {Bradley}, \citenamefont {White}, \citenamefont {Ramesh}, \citenamefont {Singh},\ and\ \citenamefont {Bernien}}]{Anand2024}%
  \BibitemOpen
  \bibfield  {author} {\bibinfo {author} {\bibfnamefont {S.}~\bibnamefont {Anand}}, \bibinfo {author} {\bibfnamefont {C.~E.}\ \bibnamefont {Bradley}}, \bibinfo {author} {\bibfnamefont {R.}~\bibnamefont {White}}, \bibinfo {author} {\bibfnamefont {V.}~\bibnamefont {Ramesh}}, \bibinfo {author} {\bibfnamefont {K.}~\bibnamefont {Singh}},\ and\ \bibinfo {author} {\bibfnamefont {H.}~\bibnamefont {Bernien}},\ }\bibfield  {title} {\bibinfo {title} {A dual-species {R}ydberg array},\ }\href@noop {} {\bibfield  {journal} {\bibinfo  {journal} {Nat. Phys.}\ }\textbf {\bibinfo {volume} {20}},\ \bibinfo {pages} {1744} (\bibinfo {year} {2024})}\BibitemShut {NoStop}%
\bibitem [{\citenamefont {Giudici}\ \emph {et~al.}(2025)\citenamefont {Giudici}, \citenamefont {Veroni}, \citenamefont {Giudice}, \citenamefont {Pichler},\ and\ \citenamefont {Zeiher}}]{Giudici2025}%
  \BibitemOpen
  \bibfield  {author} {\bibinfo {author} {\bibfnamefont {G.}~\bibnamefont {Giudici}}, \bibinfo {author} {\bibfnamefont {S.}~\bibnamefont {Veroni}}, \bibinfo {author} {\bibfnamefont {G.}~\bibnamefont {Giudice}}, \bibinfo {author} {\bibfnamefont {H.}~\bibnamefont {Pichler}},\ and\ \bibinfo {author} {\bibfnamefont {J.}~\bibnamefont {Zeiher}},\ }\bibfield  {title} {\bibinfo {title} {Fast entangling gates for {R}ydberg atoms via resonant dipole-dipole interaction},\ }\href@noop {} {\bibfield  {journal} {\bibinfo  {journal} {PRX Quantum}\ }\textbf {\bibinfo {volume} {6}},\ \bibinfo {pages} {030308} (\bibinfo {year} {2025})}\BibitemShut {NoStop}%
\bibitem [{\citenamefont {Palm}\ \emph {et~al.}(2026)\citenamefont {Palm}, \citenamefont {Li}, \citenamefont {Feng}, \citenamefont {Jürgensen},\ and\ \citenamefont {Simon}}]{Palm2026}%
  \BibitemOpen
  \bibfield  {author} {\bibinfo {author} {\bibfnamefont {L.}~\bibnamefont {Palm}}, \bibinfo {author} {\bibfnamefont {B.}~\bibnamefont {Li}}, \bibinfo {author} {\bibfnamefont {Y.~C.}\ \bibnamefont {Feng}}, \bibinfo {author} {\bibfnamefont {M.}~\bibnamefont {Jürgensen}},\ and\ \bibinfo {author} {\bibfnamefont {J.}~\bibnamefont {Simon}},\ }\bibfield  {title} {\bibinfo {title} {Enhanced {R}ydberg blockade through {RF}-tuned {F}\"orster resonance},\ }\href {https://arxiv.org/abs/2603.07958} {\bibfield  {journal} {\bibinfo  {journal} {arXiv:2603.07958}\ } (\bibinfo {year} {2026})}\BibitemShut {NoStop}%
\bibitem [{\citenamefont {Mostaan}\ \emph {et~al.}(2026)\citenamefont {Mostaan}, \citenamefont {Goswami}, \citenamefont {Schmelcher},\ and\ \citenamefont {Mukherjee}}]{Mostaan2026}%
  \BibitemOpen
  \bibfield  {author} {\bibinfo {author} {\bibfnamefont {N.}~\bibnamefont {Mostaan}}, \bibinfo {author} {\bibfnamefont {K.}~\bibnamefont {Goswami}}, \bibinfo {author} {\bibfnamefont {P.}~\bibnamefont {Schmelcher}},\ and\ \bibinfo {author} {\bibfnamefont {R.}~\bibnamefont {Mukherjee}},\ }\bibfield  {title} {\bibinfo {title} {High-fidelity non-adiabatic dark state gates for neutral atoms},\ }\href@noop {} {\bibfield  {journal} {\bibinfo  {journal} {arXiv:2602.13885}\ } (\bibinfo {year} {2026})}\BibitemShut {NoStop}%
\bibitem [{\citenamefont {Keating}\ \emph {et~al.}(2015)\citenamefont {Keating}, \citenamefont {Cook}, \citenamefont {Hankin}, \citenamefont {Jau}, \citenamefont {Biedermann},\ and\ \citenamefont {Deutsch}}]{Keating2015}%
  \BibitemOpen
  \bibfield  {author} {\bibinfo {author} {\bibfnamefont {T.}~\bibnamefont {Keating}}, \bibinfo {author} {\bibfnamefont {R.~L.}\ \bibnamefont {Cook}}, \bibinfo {author} {\bibfnamefont {A.~M.}\ \bibnamefont {Hankin}}, \bibinfo {author} {\bibfnamefont {Y.-Y.}\ \bibnamefont {Jau}}, \bibinfo {author} {\bibfnamefont {G.~W.}\ \bibnamefont {Biedermann}},\ and\ \bibinfo {author} {\bibfnamefont {I.~H.}\ \bibnamefont {Deutsch}},\ }\bibfield  {title} {\bibinfo {title} {Robust quantum logic in neutral atoms via adiabatic {R}ydberg dressing},\ }\href@noop {} {\bibfield  {journal} {\bibinfo  {journal} {Phys. Rev. A}\ }\textbf {\bibinfo {volume} {91}},\ \bibinfo {pages} {012337} (\bibinfo {year} {2015})}\BibitemShut {NoStop}%
\bibitem [{\citenamefont {Bergonzoni}\ \emph {et~al.}(2026)\citenamefont {Bergonzoni}, \citenamefont {Riso},\ and\ \citenamefont {Pupillo}}]{Bergonzoni2025}%
  \BibitemOpen
  \bibfield  {author} {\bibinfo {author} {\bibfnamefont {M.}~\bibnamefont {Bergonzoni}}, \bibinfo {author} {\bibfnamefont {R.~R.}\ \bibnamefont {Riso}},\ and\ \bibinfo {author} {\bibfnamefont {G.}~\bibnamefont {Pupillo}},\ }\bibfield  {title} {\bibinfo {title} {Fast quantum gates for neutral atoms separated by a few tens of micrometers},\ }\href {https://arXiv.org/abs/2511.20437} {\bibfield  {journal} {\bibinfo  {journal} {arXiv:2511.20437}\ } (\bibinfo {year} {2026})}\BibitemShut {NoStop}%
\bibitem [{\citenamefont {Ildefonso}\ \emph {et~al.}(2025)\citenamefont {Ildefonso}, \citenamefont {Byun}, \citenamefont {Konovalov}, \citenamefont {Kazemi}, \citenamefont {Schuler},\ and\ \citenamefont {Lechner}}]{Ildefonso2025}%
  \BibitemOpen
  \bibfield  {author} {\bibinfo {author} {\bibfnamefont {P.}~\bibnamefont {Ildefonso}}, \bibinfo {author} {\bibfnamefont {A.}~\bibnamefont {Byun}}, \bibinfo {author} {\bibfnamefont {A.}~\bibnamefont {Konovalov}}, \bibinfo {author} {\bibfnamefont {J.}~\bibnamefont {Kazemi}}, \bibinfo {author} {\bibfnamefont {M.}~\bibnamefont {Schuler}},\ and\ \bibinfo {author} {\bibfnamefont {W.}~\bibnamefont {Lechner}},\ }\bibfield  {title} {\bibinfo {title} {Expanding the neutral atom gate set: Native i{SWAP} and exchange gates from dipolar {R}ydberg interactions},\ }\href {https://arxiv.org/abs/2512.05037} {\bibfield  {journal} {\bibinfo  {journal} {arXiv:2512.05037}\ } (\bibinfo {year} {2025})}\BibitemShut {NoStop}%
\bibitem [{Nor()}]{Norrell2026SM}%
  \BibitemOpen
  \href@noop {} {}\bibinfo {note} {See Supplemental Material at [URL will be inserted by publisher] which includes references \cite{Doultsinos2025b,Jaksch2000,Saffman2020,Jandura2022,Evered2023,fan2024fast,Sibalic2017}.}\BibitemShut {Stop}%
\bibitem [{\citenamefont {Doultsinos}\ \emph {et~al.}(2025)\citenamefont {Doultsinos}, \citenamefont {Delakouras},\ and\ \citenamefont {Petrosyan}}]{Doultsinos2025b}%
  \BibitemOpen
  \bibfield  {author} {\bibinfo {author} {\bibfnamefont {G.}~\bibnamefont {Doultsinos}}, \bibinfo {author} {\bibfnamefont {A.}~\bibnamefont {Delakouras}},\ and\ \bibinfo {author} {\bibfnamefont {D.}~\bibnamefont {Petrosyan}},\ }\bibfield  {title} {\bibinfo {title} {Fundamental bound on entanglement generation between interacting {R}ydberg atoms},\ }\href@noop {} {\bibfield  {journal} {\bibinfo  {journal} {arXiv:2512.13379}\ } (\bibinfo {year} {2025})}\BibitemShut {NoStop}%
\bibitem [{\citenamefont {Wesenberg}\ \emph {et~al.}(2007)\citenamefont {Wesenberg}, \citenamefont {M\o{}lmer}, \citenamefont {Rippe},\ and\ \citenamefont {Kr\"{o}ll}}]{Wesenberg2007}%
  \BibitemOpen
  \bibfield  {author} {\bibinfo {author} {\bibfnamefont {J.~H.}\ \bibnamefont {Wesenberg}}, \bibinfo {author} {\bibfnamefont {K.}~\bibnamefont {M\o{}lmer}}, \bibinfo {author} {\bibfnamefont {L.}~\bibnamefont {Rippe}},\ and\ \bibinfo {author} {\bibfnamefont {S.}~\bibnamefont {Kr\"{o}ll}},\ }\bibfield  {title} {\bibinfo {title} {Scalable designs for quantum computing with rare-earth-ion-doped crystals},\ }\href@noop {} {\bibfield  {journal} {\bibinfo  {journal} {Phys. Rev. A}\ }\textbf {\bibinfo {volume} {75}},\ \bibinfo {pages} {012304} (\bibinfo {year} {2007})}\BibitemShut {NoStop}%
\bibitem [{\citenamefont {Chew}\ \emph {et~al.}(2022)\citenamefont {Chew}, \citenamefont {Tomita}, \citenamefont {Mahesh}, \citenamefont {Sugawa}, \citenamefont {de~L{\'e}s{\'e}leuc},\ and\ \citenamefont {Ohmori}}]{YChew2022}%
  \BibitemOpen
  \bibfield  {author} {\bibinfo {author} {\bibfnamefont {Y.}~\bibnamefont {Chew}}, \bibinfo {author} {\bibfnamefont {T.}~\bibnamefont {Tomita}}, \bibinfo {author} {\bibfnamefont {T.~P.}\ \bibnamefont {Mahesh}}, \bibinfo {author} {\bibfnamefont {S.}~\bibnamefont {Sugawa}}, \bibinfo {author} {\bibfnamefont {S.}~\bibnamefont {de~L{\'e}s{\'e}leuc}},\ and\ \bibinfo {author} {\bibfnamefont {K.}~\bibnamefont {Ohmori}},\ }\bibfield  {title} {\bibinfo {title} {Ultrafast energy exchange between two single {R}ydberg atoms on a nanosecond timescale},\ }\href@noop {} {\bibfield  {journal} {\bibinfo  {journal} {Nat. Photon.}\ }\textbf {\bibinfo {volume} {16}},\ \bibinfo {pages} {724} (\bibinfo {year} {2022})}\BibitemShut {NoStop}%
\bibitem [{\citenamefont {Saffman}\ \emph {et~al.}(2020)\citenamefont {Saffman}, \citenamefont {Beterov}, \citenamefont {Dalal}, \citenamefont {Paez},\ and\ \citenamefont {Sanders}}]{Saffman2020}%
  \BibitemOpen
  \bibfield  {author} {\bibinfo {author} {\bibfnamefont {M.}~\bibnamefont {Saffman}}, \bibinfo {author} {\bibfnamefont {I.~I.}\ \bibnamefont {Beterov}}, \bibinfo {author} {\bibfnamefont {A.}~\bibnamefont {Dalal}}, \bibinfo {author} {\bibfnamefont {E.~J.}\ \bibnamefont {Paez}},\ and\ \bibinfo {author} {\bibfnamefont {B.~C.}\ \bibnamefont {Sanders}},\ }\bibfield  {title} {\bibinfo {title} {Symmetric {R}ydberg controlled$-{Z}$ gates with adiabatic pulses},\ }\href@noop {} {\bibfield  {journal} {\bibinfo  {journal} {Phys. Rev. A}\ }\textbf {\bibinfo {volume} {101}},\ \bibinfo {pages} {062309} (\bibinfo {year} {2020})}\BibitemShut {NoStop}%
\bibitem [{\citenamefont {Fan}\ \emph {et~al.}(2024)\citenamefont {Fan}, \citenamefont {Wang},\ and\ \citenamefont {Sun}}]{fan2024fast}%
  \BibitemOpen
  \bibfield  {author} {\bibinfo {author} {\bibfnamefont {X.}~\bibnamefont {Fan}}, \bibinfo {author} {\bibfnamefont {X.}~\bibnamefont {Wang}},\ and\ \bibinfo {author} {\bibfnamefont {Y.}~\bibnamefont {Sun}},\ }\bibfield  {title} {\bibinfo {title} {Fast entangling quantum gates with almost-resonant modulated driving},\ }\href@noop {} {\bibfield  {journal} {\bibinfo  {journal} {Fundamental Research}\ } (\bibinfo {year} {2024})}\BibitemShut {NoStop}%
\bibitem [{\citenamefont {\u{S}ibali\'{c}}\ \emph {et~al.}(2017)\citenamefont {\u{S}ibali\'{c}}, \citenamefont {Pritchard}, \citenamefont {Adams},\ and\ \citenamefont {Weatherill}}]{Sibalic2017}%
  \BibitemOpen
  \bibfield  {author} {\bibinfo {author} {\bibfnamefont {N.}~\bibnamefont {\u{S}ibali\'{c}}}, \bibinfo {author} {\bibfnamefont {J.}~\bibnamefont {Pritchard}}, \bibinfo {author} {\bibfnamefont {C.}~\bibnamefont {Adams}},\ and\ \bibinfo {author} {\bibfnamefont {K.}~\bibnamefont {Weatherill}},\ }\bibfield  {title} {\bibinfo {title} {Arc: An open-source library for calculating properties of alkali {R}ydberg atoms},\ }\href@noop {} {\bibfield  {journal} {\bibinfo  {journal} {Comp. Phys. Commun.}\ }\textbf {\bibinfo {volume} {220}},\ \bibinfo {pages} {319} (\bibinfo {year} {2017})}\BibitemShut {NoStop}%
\end{thebibliography}%

\newcommand{\Hint}{H_{\mathrm{int}}}
\newcommand{\Hsfull}{\mathcal{H}_{\mathrm{full}}}
\newcommand{\PiRS}{\Pi_{rs}}
\newcommand{\etafull}{\eta_{\mathrm{full}}}
\newtheorem{theorem}{Theorem}
\newtheorem{lemma}{Lemma}
\theoremstyle{definition}
\newtheorem{assumption}{Assumption}

%\newcommand{\um}{\mu \rm{m}}

% Mark's template for an initial-tagged comment cmd
%\newcommand{\ColorComment}[3]{%
%	{\colorbox{#1}{\color{white}   \textsf{\textbf{#2}}} \textcolor{#1}{#3}}}
%  Colorful box, initials, phrase 

%\definecolor{mscolor}{rgb}{0,0.5,0.5}\newcommand{\ms}[1]{\ColorComment{mscolor}{ms}{#1}}

\setcounter{page}{1}
\setcounter{section}{0}
\setcounter{figure}{0}
\setcounter{equation}{0}
\setcounter{table}{0}

\renewcommand{\thepage}{SM.\arabic{page}}
\renewcommand{\theequation}{SM.\arabic{equation}}
\renewcommand{\thesection}{SM.\arabic{section}}
\renewcommand{\thefigure}{SM.\arabic{figure}}
\renewcommand{\thetable}{SM.\arabic{table}}

%\begin{document}

\onecolumngrid
\noindent
{ \bf Supplemental Material for \\
\Large \center Entangling gate performance and fidelity limits with neutral atom F\"orster resonances}

\rsub{\section{Rank-One Drive to Two-Eigenstate Systems}}

\rsub{Here, we consider fidelity limitations for rank one drive to two-eigenstate systems, where only a single Rydberg-Rydberg state is addressed and the eigenstate branches combine to form bright and dark states. Each atom contains Hilbert spaces divided into the Rydberg and qubit manifolds:}\\
\rsubmath{
\begin{equation}
    \mathcal{H}^A = \underbrace{\mathrm{span}\{\ket{0},\ket{1}\}}_{\mathcal{H}_q^A}\oplus\ \mathrm{span}\{\ket{a},\ket{\alpha}\},
    \qquad
    \mathcal{H}^B = \underbrace{\mathrm{span}\{\ket{0},\ket{1}\}}_{\mathcal{H}_q^B}\oplus\ \mathrm{span}\{\ket{b},\ket{\beta}\},
\end{equation}
}
\rsub{and the pair evolves under the Hamiltonian}
\rsubmath{\begin{equation}
    H(t)=\underbrace{H_A(t)\otimes\mathbb{I}+\mathbb{I}\otimes H_B(t)}_{H_{\rm loc}(t)}+H_{\rm int}, \qquad H_{\rm int}=V\left(\ket{ab}\bra{\alpha\beta}+\ket{\alpha\beta}\bra{ab}\right),
\end{equation}}
\rsub{with $\hbar=1$ as in the main text. We construct the decay observable, which is the Rydberg number operator}
\rsubmath{\begin{equation}\label{eq:proj_operator}
    \hat{\Pi}_r=Q_A\otimes\mathbb{I}+\mathbb{I}\otimes Q_B, \qquad Q_A=\ket{a}\bra{a}+\ket{\alpha}\bra{\alpha},\qquad Q_B=\ket{b}\bra{b}+\ket{\beta}\bra{\beta},
\end{equation}
such that $P_r=\bra{\psi}\hat{\Pi}_r\ket{\psi}\in[0,2]$ counts the number of excited atoms and $T_R=\int_0^{T}dt P_r(t)$ is the integrated Rydberg time. The figure of merit is therefore}
\rsubmath{\begin{equation}
    \eta\equiv V\int_0^Tdt\, P_r(t).
\end{equation}
The following analysis is a generalization of the rank one, one eigenstate bound of \cite{Doultsinos2025b}. The extra Rydberg states require careful accounting resulting in a more complicated proof structure.\\
\begin{assumption}[rank-one drive]\label{assumption:rank_1_drive}
     For every $t$, $H_A(t)$ leaves $\mathcal{H}_q^A\oplus \mathbb{C}\ket{a}$ invariant. Likewise $H_B(t)$ leaves $\mathcal{H}_q^B\oplus \mathbb{C}\ket{b}$. Additionally $\mathbb{C}\ket{\alpha}$ and $\mathbb{C}\ket{\beta}$ are invariant under $H_{A}$ and $H_{B}$ as well.
\end{assumption}}
\rsub{\begin{assumption}[boundary conditions]\label{assumption:boundary_conditions}
    $\ket{\psi(0)}\in\mathcal{H}_q^A\otimes \mathcal{H}_q^B$ is a product state, and $\ket{\psi(T)}\in\mathcal{H}_q^A\otimes \mathcal{H}_q^B$ is maximally entangled.
\end{assumption}}

\rsub{
\begin{theorem}\label{theorem:rank_1_two_estsate}
    Under Assumptions \ref{assumption:rank_1_drive} and \ref{assumption:boundary_conditions},
    \begin{equation}
        \eta\geq 1+\frac{\pi}{2},
    \end{equation}
    and consequently the fidelity any entangling gate can achieve is bounded by
    \begin{equation}
        1-\mathcal{F}\geq \frac{1+\pi/2}{V\tau_R},
    \end{equation}
    for interaction strength $V$ and Rydberg lifetime $\tau_R$.
\end{theorem}}
\rsub{We construct a series of Lemmas to aid us in proving Theorem \ref{theorem:rank_1_two_estsate}.}
\rsub{\begin{lemma}[invariant subspace]\label{lemma:invariant_subspace}
    Let
    \begin{equation}
        \mathcal{H}_\varphi=\mathrm{span}\{\ket{0},\ket{1},\ket{a}\}_A\otimes\mathrm{span}\{\ket{0},\ket{1},\ket{b}\}_B, \qquad \mathcal{S}=\mathcal{H}_\varphi \oplus \mathbb{C}\ket{\alpha\beta}.
    \end{equation}
    Under Assumption \ref{assumption:rank_1_drive}, $H(t)\mathcal{S}\subseteq\mathcal{S}$ for all t. Since $\ket{\psi(0)}\in\mathcal{H}_q^A\otimes \mathcal{H}_q^B\subset\mathcal{S}$, the trajectory satisfies $\ket{\psi(t)}\in \mathcal{S}\ \forall\ t\in[0,T]$.\end{lemma}
\begin{proof}
    With $\ket{\varphi}\in \mathcal{H}_\varphi$ and $c\in \mathbb{C}$, we consider the local and exchange dynamics separately.\\
    \textit{Local.} By Assumption \ref{assumption:rank_1_drive}, $H_A$ and $H_B$ preserve their local subspaces containing $\ket{0},\ket{1},$ and $\ket{a}$ or $\ket{b}$ respectively. Hence, $H_{\rm loc}$ preserves $\mathcal{H}_\varphi$. \\
    \textit{Exchange.} Since $\bra{\alpha\beta}\ket{\varphi}=0, \ H_{\rm int}\ket{\varphi}=V\expval{ab |\varphi}\ket{\alpha\beta}\in\mathbb{C}\ket{\alpha \beta}$, and $H_{\rm int}\ket{\alpha\beta}= V\ket{ab}\in \mathcal{H}_\varphi$.\\
    Therefore, $H\left( \ket{\varphi}+c\ket{\alpha\beta}\right)\in \mathcal{S}$, and any trajectory starting in $S$ must remain in $S$.
\end{proof}
Thus the states $\ket{a\beta}$ and $\ket{\alpha b}$ remain entirely unpopulated and at all times $t$ we can write
\begin{equation}\label{eq:phi_decomposition}
    \ket{\psi(t)}=\ket{\varphi(t)}+z(t)\ket{\alpha\beta}, \qquad \ket{\varphi(t)}\in\mathcal{H}_\varphi, \qquad \varepsilon(t) \equiv \abs{z(t)}.
\end{equation}
}
\rsub{
\begin{lemma}[block structure] \label{lemma:block_structure}
    With the decomposition presented in (\ref{eq:phi_decomposition}),
    \begin{equation}
        \rho_A=\mathrm{Tr_B}\ket{\psi}\bra{\psi}=\rho_A^\varphi\oplus \varepsilon^2\ket{\alpha}\bra{\alpha},
    \end{equation}
    where $\rho_A^\varphi=\mathrm{Tr_B}\ket{\varphi}\bra{\varphi}$. Consequently $\ket{\psi}$ has at most four non-zero Schmidt coefficients: $\varepsilon$ with Schmidt vector $\ket{\alpha \beta}$ and $d_1\geq d_2\geq d_3\geq0$ with Schmidt vectors in $\mathcal{H}_\varphi$. Moreover $\sum_id_i^2 +\varepsilon^2=1$.
\end{lemma}
\begin{proof}
    Expand $\mathrm{Tr_B}\ket{\psi}\bra{\psi}=\rho_A^\varphi+\varepsilon^2\ket{\alpha}\bra{\alpha}+z\mathrm{Tr_B}\ket{\alpha\beta}\bra{\varphi}+\mathrm{h.c.}$. The cross terms are zero as $\expval{\alpha\beta|\varphi}=0$ by orthogonality. The reduced denisty matrix therefore has a direct-sum form, and the $3\times3$ block structure contributes at most three Schmidt coefficients.
\end{proof}
}

\rsub{Pick Schmidt vectors $\ket{u_i}$ and $\ket{v_i}$ for $\ket{\varphi}$, such that the following holds
\begin{equation}
    \ket{\psi}=\sum_{i=1}^3d_i\ket{u_iv_i}+ z\ket{\alpha\beta}, \qquad \sum_{i=1}^3d_i^2 + \varepsilon^2=1.
\end{equation}
The sets $\{\ket{u_{i}}\}^{3}_{i=1}$ and $\{\ket{v_{i}}\}^{3}_{i=1}$
 therefore are orthonormal bases for $\mathrm{span}\{\ket{0},\ket{1},\ket{a}\}$ and $\mathrm{span}\{\ket{0},\ket{1},\ket{b}\}$, respectively.}

\rsub{\begin{lemma}[entanglement rate]\label{lemma:entanglement_rate}
For almost every $t\in[0,T]$, there exists non-negative overlap factors $w_1,w_2,w_3$ such that 
\begin{align}
    \label{eq:transport_bound}\abs{\dot{d_i}}&\leq V\varepsilon w_i \qquad (i=1,2,3)\\
    \sum_iw_i&\leq 1\\
     \label{eq:cost_bound}P_r&\geq 2\sum_iw_id_i^2+2\varepsilon^2.
\end{align}
\end{lemma}
\begin{proof}
    Consider an eigenvalue branch with $d_i>0$. Since $d_i^2=\bra{u_i}\rho_A\ket{u_i}\equiv \lambda_i$, differentiation gives
    \begin{equation}
        \dot{\lambda}_i=\bra{u_i}\dot{\rho}_A\ket{u_i}.
    \end{equation}
    With $\dot{\rho}_A=-i\mathrm{Tr_B}[H,\ket{\psi}\bra{\psi}]$,
    \begin{equation}
        \dot{\lambda}_i=2d_i\mathrm{Im}\bra{u_iv_i}H\ket{\psi}.
    \end{equation}
    Therefore,
    \begin{equation}
        \dot{d}_i=\mathrm{Im}\bra{u_iv_i}H\ket{\psi}.
    \end{equation}
    Considering local controls
    \begin{equation}
        \bra{u_iv_i}H_{\mathrm{loc}}\ket{\psi}=d_i\left( \bra{u_i}H_A\ket{u_i} + \bra{v_i}H_B\ket{v_i}\right)\in \mathbb{R},
    \end{equation}
    so local controls do not change Schmidt coefficients. For the interaction term,
    \begin{equation}
        \bra{u_iv_i}H_{\mathrm{int}}\ket{\psi}=Va_ib_iz, \qquad a_i\equiv\expval{u_i|a}, \qquad b_i\equiv\expval{v_i|b}.
    \end{equation}
    Thus, with $w_i\equiv |a_i||b_i|$,
    \begin{equation}
        |\dot{d}_i|\leq V\varepsilon w_i.
    \end{equation}
    Due to the orthonormality of the Schmidt vectors,
    \begin{equation}
        \sum_i|a_i|^2=\sum_i|b_i|^2=1.
    \end{equation}
    The Cauchy-Schwarz inequality therefore provides
    \begin{equation}
        \sum_i w_i\sum_i|a_i||b_i|\leq 1.
    \end{equation}
    Finally, directly evaluating $P_r$ gives
    \begin{equation}
        P_r=\sum_i d_i^2(|a_i|^2 +|b_i|^2)+2\varepsilon^2\geq 2\sum_i w_id_i^2 + 2\varepsilon^2,
    \end{equation}
    where the inequality follows from $|a_i|^2 + |b_i|^2\geq 2|a_i||b_i|$.
\end{proof}}

\rsub{Let
\begin{equation}
    \mathcal{D}=\{d=(d_1,d_2,d_3)\ : \ d_1\geq0, \sum_id_i^2\le1\}
\end{equation} 
and 
\begin{equation}
    \varepsilon(d)=\left(1-\sum_i d_i^2\right)^{1/2}.
\end{equation}
\begin{lemma}[pointwise reduction]\label{lemma:pointwise_reduction}
Let $W$ be differentiable at $d\in\mathcal{D}$. The inequality 
\begin{equation}\label{eq:pointwise_lemma_claim}
    \varepsilon\sum_iw_i\abs{\frac{\partial W}{\partial d_i}}\leq2\sum_i w_id_i^2+2\varepsilon^2
\end{equation}
holds for all $w_i\geq 0$ with $\sum_iw_i\leq1$ if and only if
\begin{equation}\label{eq:W-condition}
    \varepsilon\abs{\frac{\partial W}{\partial d_i}}\leq 2d_i^2 + 2\varepsilon^2 \qquad \text{for every i}.
\end{equation}
\end{lemma}
\begin{proof}
   ($\Rightarrow$) Choose $w$ to be the $i$-th vector $w_i=1$ and all others zero. Eq. (\ref{eq:pointwise_lemma_claim}) reduces to Eq. (\ref{eq:W-condition}) exactly.\\
   ($\Leftarrow$) Multiply the $i$-th inequality (Eq. (\ref{eq:W-condition})) by $w_i\geq 0$ and sum over $i$:
   \begin{equation}
       \varepsilon\sum_iw_i\abs{\frac{\partial W}{\partial d_i}}\leq 2\sum_iw_id_i^2+2\varepsilon^2\sum_iw_i\leq2\sum_iw_id_i^2 + 2 \varepsilon^2.
   \end{equation}
\end{proof}
}

\rsub{Define $G\ :\ [0,1]\mapsto \mathbb{R}$ and $W\ : \ \mathcal{D}\mapsto \mathbb{R}$ as
\begin{align}
    G(x)&=\begin{cases}
    2x^2, & 0\le x\le \tfrac{1}{\sqrt2},\\
    1+\dfrac{\pi}{2}-2\arccos x,
      & \tfrac{1}{\sqrt2}\le x\le 1,
  \end{cases}\\
  W(d) &= -G(d_1)+2d_2^2 \label{eq:W}
\end{align}
The two piecewise components of $G$ agree in value and first derivative at $x=\tfrac{1}{\sqrt2}$ ($G\left(\tfrac{1}{\sqrt2}\right) = 1$ and $G'\left(\tfrac{1}{\sqrt2}\right)=2\sqrt{2}$).
\begin{lemma}[comparison bound]\label{lemma:comparison-bound}
    The function $W(d)$ (Eq. (\ref{eq:W})) satisfies Eq. (\ref{eq:W-condition}) at every $d\in\mathcal{D}$, with $d_1<1$.
\end{lemma}
\begin{proof}
    We require normalization: $\varepsilon^2=1 -d_1^2-d_2^2-d_3^2\leq1-d_1^2$. Therefore 
    \begin{equation}
        \varepsilon\leq\sqrt{1-d_1^2}\equiv m. \label{eq:epsilon_leq_m}
    \end{equation}
    There are three coordinates, and two branches in $G$ for $i=1$, making four cases total.\\
    \textit{Case 1:} $i=3$. $W$ does not depend on $d_3$, so $\tfrac{\partial W}{\partial d_3}=0$, and Eq. (\ref{eq:W-condition}) reads $0\leq2d_3^2+2\varepsilon^2$, which holds.\\
    \textit{Case 2:} $i=2$. $\tfrac{\partial W}{\partial d_2}=\abs{4d_2}=4d_2$, so we must show $4d_2\varepsilon\leq 2d_2^2+2\varepsilon^2$. Rearranged, this reads $0\leq2d_2^2-4d_2\varepsilon+2\varepsilon^2=2(d_2-\varepsilon)^2$, which holds. \\
    \textit{Case 3:} $i=1$ and $d_1\leq\tfrac{1}{\sqrt{2}}$. $\tfrac{\partial W}{\partial d_1}=G'(d_1)=4d_1$, and the same logic from Case 2 holds: $4d_1\varepsilon\leq 2d_1^2 + 2\varepsilon^2\iff 0\leq 2(d_1-\varepsilon)^2$.\\
    \textit{Case 4:} $i=1$ and $d_1\geq\tfrac{1}{\sqrt{2}}$. $\tfrac{\partial W}{\partial d_1}=G'(d_1)=2/m$. We must show
    \begin{equation}
        \frac{2\varepsilon}{m}\leq 2d_1^2+2\varepsilon^2.
    \end{equation}
    dividing by $2\varepsilon$, an alternative form is:
    \begin{equation}
        \frac{1}{m}\leq \frac{d_1^2}{\varepsilon}+\varepsilon\equiv f(\varepsilon)
    \end{equation}
    $f'(\varepsilon)=1-d_1^2/\varepsilon^{2}$, which is negative for $\varepsilon<d_1$, so $f$ is decreasing on the interval $(0,d_1]$. $\varepsilon$ must exist on this interval for two reasons: (1) $m\leq d_1$ because $d_1\geq \tfrac{1}{\sqrt{2}}$ implies $m^2=1-d_1^2\leq\tfrac{1}{2}\leq d_1^2$, and (2) $\varepsilon\leq m$ by Eq. (\ref{eq:epsilon_leq_m}). Hence, $0<\varepsilon\leq m \leq d_1$, and since $f$ is decreasing
    \begin{equation}
        f(\varepsilon)\geq f(m)=\frac{d_1^2}{m}+m=\frac{1}{m},
    \end{equation}
    which is what was required.
\end{proof}
}

\rsub{\begin{lemma}[regularity of $W$]\label{lemma:regularity_of_W}
    Let $\chi(t)\equiv\arccos d_1(t)$. $\chi$ is Lipshitz on $[0,T]$ with constant $V$, that is, $\abs{\chi(t_2)-\chi(t_1)}\leq V\abs{t_2-t_1}$ for all $t_1,t_2$. Furthermore, $t\mapsto W(d(t))$ is Lipshitz and therefore absolutely continuous.   
\end{lemma}
\begin{proof}
    Whenever $d_1<1$, we may differentiate $\chi$. We then use, in order presented the transport bound (Eq. (\ref{eq:transport_bound})), the fact that $w_1\leq\sum_i w_i\leq 1$, and Eq. (\ref{eq:epsilon_leq_m}):
    \begin{equation}\label{eq:chi_dot_bound}
        \abs{\dot\chi}=\frac{\abs{\dot{d_1}}}{\sqrt{1-d_1^2}}\leq\frac{V\varepsilon w_1}{\sqrt{1-d_1^2}}\leq\frac{V\varepsilon}{\sqrt{1-d_1^2}}\leq V.
    \end{equation}
    So $\chi$ is bounded when $d_1<1$. For $d_1=1$, where $\chi=0$, consider the set $U=\{t  : d_1(t)<1\}$. $U$ is open and a countable union of disjoint intervals. By continuity and Eq. (\ref{eq:chi_dot_bound}) $\chi$ is Lipshitz. For $t_1<t_2$ in different components of $U$, there exists $\tau \in (t_1,t_2)$ with $d_1(\tau)=1$, and we can write
    \begin{equation}
        \abs{\chi(t_1)-\chi(t_2)}\leq \chi(t_1)+\chi(t_2)=\abs{\chi(t_1)-\chi(\tau)}+\abs{\chi(\tau)-\chi(t_2)}\leq V(\tau-t_1)+V(t_2-\tau)=V(t_2-t_1).
    \end{equation}
    Hence, $\chi$ is V-Lipshitz on $[0,T]$.\\
    For $W$, on the closed set $\{t \ : \ d_1(t)\geq\tfrac{1}{\sqrt{2}}\}$, $W(d(t))=2\chi(t)-(1+\tfrac{\pi}{2})+2d_2(t)^2$. On the closed set $\{t \ : \ d_1(t)\leq\tfrac{1}{\sqrt{2}}\}$, $W(d(t))=-2d_1(t)^2+2d_2(t)^2$. Both expressions are Lipshitz in $t$, and so $W(d(t))$ is Lipshitz.
\end{proof}
}
\rsub{We can now turn to proving Theorem \ref{theorem:rank_1_two_estsate}.
\begin{proof}
    The preceeding lemmas imply, for almost all $t$ (up to limiting set of times of measure-zero)
    \begin{equation}\label{eq:dW_dt_ineq}
        \frac{d}{dt}W({\mathbf{d}}(t))\leq VP_r(t), \qquad \mathbf{d}=(d_1,d_2,d_3).
    \end{equation}
    To see this explicitly, choose a time at which the Lipshitz functions and bounds in Lemma \ref{lemma:entanglement_rate} are differentiable. If $d_1=1$, normalization requires $\chi=d_2=d_2=\epsilon=0$. Because $\chi$ and $d_2$ are non-negative, their derivatives vanish, and $W=-G(\chi)+2d_2^2$ from Lemma \ref{lemma:regularity_of_W} provides $\dot{W}=0\leq VP_r$.\\
    If $d_1<1$, we differentiate and chain four inequalities
    \begin{equation}
      \frac{dW}{dt}
      \;\overset{(1)}{=}\; \sum_i \frac{\partial W}{\partial d_i}\dot d_i
      \;\overset{(2)}{\le}\; \sum_i \Big|\frac{\partial W}{\partial d_i}\Big|\,|\dot d_i|
      \;\overset{(3)}{\le}\; V\varepsilon\sum_i w_i\Big|\frac{\partial W}{\partial d_i}\Big|
      \;\overset{(4)}{\le}\; V\Big(2\sum_i w_i d_i^2+2\varepsilon^2\Big)
      \;\overset{(5)}{\le}\; V P_r ,
    \end{equation}
    where (1) is the chain rule; (2) bounds a sum by the sum of absolute values; (3) is the transport bound from Lemma \ref{lemma:entanglement_rate} (Eq. (\ref{eq:transport_bound})); (4) combines Lemma \ref{lemma:comparison-bound} with  Lemma \ref{lemma:pointwise_reduction}; and (5) is the cost bound from Lemma \ref{lemma:entanglement_rate} (Eq. (\ref{eq:cost_bound})).
    Thus Eq. (\ref{eq:dW_dt_ineq}) holds almost everywhere. \\
    By Lemma \ref{lemma:regularity_of_W}, $W(\mathbf{d}(t))$ is absolutely continuous, so we can integrate the inequality
    \begin{equation}\label{eq:eta_of_W}
        \eta=V\int_0^TP_r(t)dt\geq W(\mathbf{d}(T))-W(\mathbf{d}(0)).
    \end{equation}
    The initial state has $\mathbf{d}(0)=(1,0,0)$ and $\varepsilon(0)=0$, hence 
    \begin{equation}
        W(\mathbf{d}(0))=-G(1) = -\left(1+\frac{\pi}{2}\right).
    \end{equation}
    The final maximally entangled qubit state has $\mathbf{d}(T)=(\tfrac{1}{\sqrt{2}},\frac{1}{\sqrt{2}},0)$ and $\varepsilon(T)=0$, hence
    \begin{equation}
        W(\mathbf{d}(T))=-1+2\left(\frac{1}{\sqrt{2}}\right)^2=0.
    \end{equation}
    Substituting these values into the right-hand side of Eq. (\ref{eq:eta_of_W}) yields
    \begin{equation}
        \eta \geq 1+\frac{\pi}{2}.
    \end{equation}
\end{proof}
}
\rsub{The same bound as was determined for one-eigenstate models in \cite{Doultsinos2025b}. We note, though, that for the other configurations considered (rank-one drive to one-eigenstate models and rank-two drive to two-eigenstate models) the corresponding bound is a tight bound; that is, there exist pulses that exactly saturate the bound. The above proof is not a tight bound and as such, the tight bounding coefficient may be greater than $1+\frac{\pi}{2}$.}\\

\rsub{\section{Rank-Two Drive to Two-Eigenstate Systems}}
    
We consider the Förster exchange Hamiltonian
\begin{equation}
\Hint = V\left(\ket{ab}\bra{\alpha\beta}+\ket{\alpha\beta}\bra{ab}\right)
\label{eq:Hint_note}
\end{equation}
and the decay observable
\begin{equation}
\Pi=(\ket a\bra a+\ket\alpha\bra\alpha)_A+(\ket b\bra b+\ket\beta\bra\beta)_B.
\label{eq:Pi_note}
\end{equation}

We consider the subspace, \rsub{with $\ket{g}\in \mathrm{span}\{\ket{0},\ket{1}\}$} 
$$
\Hsfull=\mathrm{span}\{\ket g,\ket a,\ket\alpha\}_A\otimes \mathrm{span}\{\ket g,\ket b,\ket\beta\}_B.
$$

Let
$$
|\psi\rangle=\sum_i c_i |u_i v_i\rangle
$$
be a Schmidt decomposition, with the coefficients ordered so that
\(c_1=c_{\max}\) is the largest Schmidt coefficient. We fix a
min-entropy value
$$
S_{\min}(\psi)=s,\qquad 0<s<1.
$$
Since
$$
S_{\min}(\psi)=-\log_2(c_1^2),
$$
we have
$$
x:=c_1^2=2^{-s}\in(1/2,1).
$$
The interval \(0<s<1\) is the interior of the path from a product
state, for which \(S_{\min}=0\) and \(c_1^2=1\), to a one-ebit
target, for which \(S_{\min}=1\) and \(c_1^2=1/2\). We exclude
the endpoints in the pointwise argument because at \(s=0\) there is
no nonzero Schmidt remainder to normalize, while at \(s=1\) the
largest Schmidt coefficient may become degenerate.

For \(0<s<1\), however, \(x>1/2\). Since the squared Schmidt
coefficients sum to one, the remaining Schmidt weight is
$$
\sum_{i>1} c_i^2 = 1-x < x,
$$
so no other Schmidt coefficient can equal \(c_1\). Thus the largest
Schmidt branch is simple, and we may unambiguously set
$$
|u\rangle:=|u_1\rangle,\qquad |v\rangle:=|v_1\rangle.
$$
To isolate the dominant product branch from the rest of the state,
define
\begin{equation}
|n\rangle
= \frac{1}{\sqrt{1-x}}
  \sum_{i>1}
  c_i |u_i v_i\rangle .
\label{eq:n_note}
\end{equation}
Then \(\|n\|=1\), and the state can be written as
\begin{equation}
|\psi\rangle
=
\sqrt{x}\,|uv\rangle
+
\sqrt{1-x}\,|n\rangle .
\label{eq:psi_note}
\end{equation}

Moreover, because the Schmidt vectors are orthonormal, for every
\(i>1\) we have
$$
|u_i\rangle\in u^\perp,
\qquad
|v_i\rangle\in v^\perp.
$$
Therefore
$$
|n\rangle\in u^\perp\otimes v^\perp.
$$
Equivalently, \(|n\rangle\) has no component whose \(A\)-side is
\(|u\rangle\) and no component whose \(B\)-side is \(|v\rangle\):
for arbitrary local vectors \(|w_A\rangle\) and \(|w_B\rangle\),
$$
\langle u w_B|n\rangle=0,
\qquad
\langle w_A v|n\rangle=0.
$$
This stronger orthogonality, not merely$\bra{uv}\ket{n} = 0$.

\rsub{Using the Rydberg number operator defined in Eq. (\ref{eq:proj_operator}), d}efine
\begin{equation}
P_1 := \langle uv|\rsub{\hat{\Pi}_{r}}|uv\rangle,
\quad
P_n := \langle n|\rsub{\hat{\Pi}_{r}}|n\rangle
\quad\rsub{P(\psi):=\langle \psi|\hat{\Pi}_r|\psi\rangle.}
\end{equation}

Since $|n\rangle\in u^\perp\otimes v^\perp$, we have
$$
\langle uv|\rsub{\hat{\Pi}_{r}}|n\rangle
=
\langle Q_A u\otimes v|n\rangle+\langle u\otimes Q_B v|n\rangle=0,
$$
so
\begin{equation}
P(\psi)=xP_1+(1-x)P_n.
\label{eq:P_split_note}
\end{equation}

For $0<s<1$ the largest Schmidt coefficient is simple, and Ref. \cite{Doultsinos2025b} implies the rate of change of the minimum entropy goes as
\begin{equation}
\dot S_{\min}(\psi)
=
-\frac{2}{\ln 2}\frac{1}{c_1}
\sum_{i>1} c_i\,\Im\bra{u_1 v_1}\Hint\ket{u_i v_i}
=
-\frac{2}{\ln 2}\sqrt{\frac{1-x}{x}}\;\Im\langle uv|\Hint|n\rangle.
\label{eq:Sdot_note}
\end{equation}
Therefore
\begin{equation}
|\dot S_{\min}(\psi)|
\le \frac{2}{\ln 2}\sqrt{\frac{1-x}{x}}\;|\langle uv|\Hint|n\rangle|.
\label{eq:Sdot_abs_note}
\end{equation}

Write
$$
\mu:=\langle ab|uv\rangle = u_a v_b,
\qquad
\nu:=\langle\alpha\beta|uv\rangle = u_\alpha v_\beta.
$$
Then 
\begin{equation}
\begin{aligned}
\langle uv|\Hint|n\rangle
&= V\bigl(\mu^*\,\langle\alpha\beta|n\rangle + \nu^*\,\langle ab|n\rangle\bigr), \\
\bigl|\langle uv|\Hint|n\rangle\bigr|
&\le V\sqrt{\bigl(|\mu|^2+|\nu|^2\bigr)
         \bigl(|\langle\alpha\beta|n\rangle|^2 + |\langle ab|n\rangle|^2\bigr)},
\end{aligned}
\end{equation}
where the second line uses the Cauchy--Schwarz inequality, $|\langle u,v\rangle|^2 \le \langle u,u\rangle\langle v,v\rangle$.

Bound $|\mu|^2+|\nu|^2$ by setting
$$
p:=|u_a|^2+|u_\alpha|^2,
\quad
q:=|v_b|^2+|v_\beta|^2,
\quad
P_1=p+q.
$$
Then
$$
|\mu|^2+|\nu|^2
= |u_a|^2|v_b|^2 + |u_\alpha|^2|v_\beta|^2
\le (|u_a|^2+|u_\alpha|^2)(|v_b|^2+|v_\beta|^2)
= pq
\le \frac{(p+q)^2}{4}
= \frac{P_1^2}{4}.
$$

Now expand $|n\rangle$ in the local product basis. Writing
\begin{align}
|u_i\rangle &= u_{i,g}|g\rangle + u_{i,a}|a\rangle + u_{i,\alpha}|\alpha\rangle, \\
|v_i\rangle &= v_{i,g}|g\rangle + v_{i,b}|b\rangle + v_{i,\beta}|\beta\rangle ,
\label{change_of_basis}
\end{align}
we obtain
\begin{equation}
|n\rangle
= \frac{1}{\sqrt{1-x}} \sum_{i>1}
  \sum_{j\in\{g,a,\alpha\}}\sum_{k\in\{g,b,\beta\}}
  c_i\,u_{i,j}\,v_{i,k}\,|jk\rangle.
\end{equation}

Define
$$
C_{jk}
:= \frac{1}{\sqrt{1-x}}\sum_{i>1} c_i\,u_{i,j}\,v_{i,k},
\quad
|n\rangle = \sum_{j,k} C_{jk}|jk\rangle.
$$

Each product state $|jk\rangle$ is an eigenvector of $\rsub{\hat{\Pi}_{r}}$ with eigenvalue $0$, $1$, or $2$, and the eigenvalue equals $2$ exactly when $j\in\{a,\alpha\}$ and $k\in\{b,\beta\}$. Therefore
$$
P_n = \langle n|\rsub{\hat{\Pi}_{r}}|n\rangle \ge 2\bigl(|C_{ab}|^2 + |C_{\alpha\beta}|^2\bigr)
= 2\bigl(|\langle ab|n\rangle|^2 + |\langle\alpha\beta|n\rangle|^2\bigr).
$$

Combining we obtain
\begin{equation}
\bigl|\langle uv|\Hint|n\rangle\bigr|
\le V\,\frac{P_1}{2}\sqrt{\frac{P_n}{2}}.
\label{eq:matrix_bound_note}
\end{equation}

If $P_1=0$ or $P_n=0$, then \eqref{eq:matrix_bound_note} gives $\langle uv|\Hint|n\rangle = 0$, hence $\dot S_{\min}(\psi)=0$ by \eqref{eq:Sdot_note}, and the desired lower bound is trivial. We therefore assume $P_1>0$ and $P_n>0$.

Combining \eqref{eq:Sdot_abs_note}, \eqref{eq:matrix_bound_note}, and \eqref{eq:P_split_note},
\begin{equation}
\frac{P(\psi)}{|\dot S_{\min}(\psi)|}
\ge \frac{\ln 2}{V}\sqrt{\frac{x}{1-x}}\;
   \frac{xP_1 + (1-x)P_n}{P_1\sqrt{P_n/2}}.
\label{eq:ratio_note}
\end{equation}

Now set $t=\sqrt{P_n/2}\in(0,1]$. Then
$$
\frac{x P_1+(1-x)P_n}{P_1\sqrt{P_n/2}}
=
\frac{x}{t}+\frac{2(1-x)}{P_1}\,t
\ge
\frac{x}{t}+(1-x)t,
$$
because $P_1 \le 2$. Finally,
\begin{equation}
\frac{x}{t}+(1-x)t-1
=
\frac{(1-t)\bigl(x-(1-x)t\bigr)}{t}
\ge 0,
\label{eq:scalar_closure_audited}
\end{equation}
since $0<t\le 1$ and $x\ge 1/2$. Therefore
$$
\frac{x P_1+(1-x)P_n}{P_1\sqrt{P_n/2}}\ge 1,
$$
which implies from \eqref{eq:ratio_note} that
\begin{equation}
\frac{P(\psi)}{|\dot S_{\min}(\psi)|}
\ge
\frac{\ln 2}{V}\sqrt{\frac{x}{1-x}}.
\label{eq:pointwise_note}
\end{equation}

This lower bound is saturated by
$$
|\psi_s\rangle = \sqrt{x}\,|ab\rangle + i\,\operatorname{sgn}(V)\sqrt{1-x}\,|\alpha\beta\rangle,
$$
for which $P_1=P_n=2$ and $|\langle uv|\Hint|n\rangle|=V$.
Hence
\begin{equation}
G_{\mathrm{full}}(s) \equiv \min_{\psi_s}
\frac{P(\psi_s)}{|\dot S_{\min}(\psi_s)|}
= \frac{\ln 2}{V}\frac{1}{\sqrt{2^s - 1}}.
\label{eq:GofS}
\end{equation}
The proof above is for $0<s<1$; the value at $s=1$ follows by continuity from below and is attained by the wavefunction $ \ket{\psi}=(|ab\rangle + i\,\operatorname{sgn}(V)|\alpha\beta\rangle)/\sqrt{2}$.

Integrating $G_{\mathrm{full}}(s)$ over the interval $s\in[0,1]$ yields
\begin{equation}
\etafull = V\int_0^1 G_{\mathrm{full}}(s)\,ds = \frac{\pi}{2}.
\end{equation}

The proof presented above relies on rank-two coupling to the Rydberg manifold. A two-eigenstate model with only rank-one coupling would not be able to map the entanglement generated from the exchange interaction back to the qubit manifold. Therefore, the increased fidelity bound depends on the \textit{rank} of the pulse sequence, not the dimensionality of the atomic model. Since the rank of the pulse must be no greater than the dimensionality of the atomic model, though, higher-dimensional models are required to accurately evaluate pulses of large rank. \\

\section{Physical Realization}
\label{sec:sm:physical-realization}
\subsection{Rank-Two Pulses}
Since the single-atom target states $\ket{a}$ and $\ket{b}$ are  spectrally isolated from the exchange states $\ket{\alpha}$ and $\ket{\beta}$, rank-two access to the Rydberg pair state manifold will require one laser per single-atom state. In symmetric structures, where both atoms are the same species, $\ket{a}$ and $\ket{b}$ are identical states, so they share an excitation laser. 
Functionally, to realize rank-two control, one would have to add lasers that address $\ket{\alpha}$ and $\ket{\beta}$ in addition to the existing lasers that excite $\ket{a}$ and $\ket{b}$.
This is most-easily achieved by coupling to Rydberg states of different values of $\ell$, for instance $np$ and $ns$ states, where, for the case of an $s$ orbital ground state,  the $p$ state is coupled via a one-photon excitation and two-photon excitation to the $s$ state.

\subsection{F\"orster Resonances and Interaction Strength}
Here, we consider the needed interaction strength $V$ to achieve the gate fidelities presented in the main text.
Since we have identified pulses nearly saturating the gate fidelity bound, we  calculate interaction strength as given by the fundamental gate fidelity bound, shown in Fig. \ref{fig:V_requirement}.
From this figure, we can see that in order to generate entanglement with fidelity $\mathcal{F}>0.999$, an interaction strength of at least 2-5 MHz is required, depending on the lifetime of the specific state addressed. 

\begin{figure}
    \centering
    \includegraphics[width=0.5\linewidth]{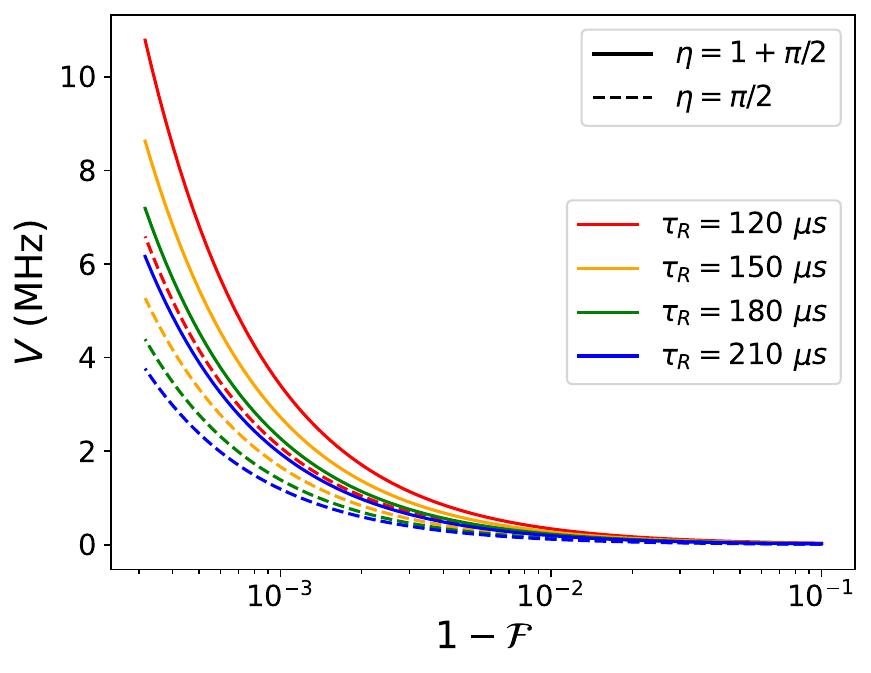}
    \caption{Interaction energy required to generate entanglement at a given fidelity.}
    \label{fig:V_requirement}
\end{figure}

F\"orster resonances, known for their favorable scaling properties, can provide interaction strengths of this magnitude at impressively large atom separation distances. The resonance between Cs atoms $\ket{64s_{1/2}} \otimes \ket{68s_{1/2}}$, for example, provides 2.88 MHz at $11\um$, thereby providing the possibility of an entangling gate of up to 4 atoms away in a tightly spaced lattice at moderate values of $n$.  

Additional resonances at higher principal quantum number could be used to increase this spacing or afford higher gate fidelity at the same spacing. Further, higher $n$ states provide longer Rydberg lifetimes, thereby further increasing the fundamental fidelity limit. 

For rank-two pulses, this distance can be extended as the fundamental limit is approximately 40 percent lower than for rank-one pulses. For the resonance considered above, this translates to an atom separation distance of $\sim13\ \um$, following $1/r^3$ scaling of resonant dipole-dipole interactions.

\section{Existing Gate Protocols with F\"orster Resonances}
\label{sec:sm:existing_protocols}
Here, we provide details on the gate pulses used to calculate gate fidelities using gate pulses commonly used in neutral atom quantum processors. 
\subsection{$\bf \pi-2\pi-\pi$}

For the original $\pi-2\pi-\pi$ Rydberg gate \cite{Jaksch2000} we assume degree-6 super-Gaussian pulse profiles for each pulse. Explicitly, the functional form of each pulse $i$ for this gate is 
\begin{equation}
    \Omega_i(t; t_0, \tau)= \Omega_{\rm max}e^{-((t-t_0)/\tau)^6},
\end{equation}
for $\Omega_c(t)$ and $\Omega_t(t)$ defined as the Rabi rates on the control and target qubits, respectively. The three pulses are timed such that pulse overlap is $0.1\%$ between the $\pi$ and $2\pi$ tails and the same percentage sets the start and end times for the gate.

\subsection{Adiabatic Rapid Passage (ARP)}
The parameters that describe the Rabi rate ($\Omega$) and detuning ($\Delta$) of this pulse are \cite{Saffman2020}
\begin{equation}
    \Omega(t)=\Omega_{\rm max}\,\frac{e^{-(t-t_0)^4/\tau^4}-a}{1-a},
\qquad
a=\exp\!\left[-\left(\frac{t_0}{\tau}\right)^4\right],
\end{equation}
\begin{equation}
\Delta_r(t)=-\Delta_r\cos\!\left(\frac{2\pi}{T}\,t\right),
\qquad t\in[0,T].
\end{equation}

\subsection{Time Optimal (TO)}

TO gates \cite{Jandura2022} have a pulse profile with uniform Rabi frequency and a time-varying phase. For this phase, parameterize the phase with an assumed functional form
\begin{equation}
    \Phi(t)=A\cos\left[\frac{\omega}{\Omega}\left(t-\frac{T}{2}\right)-\phi\right]+\frac{\delta}{\Omega}t,
\end{equation}
which approximates the optimal phase function for a total gate time $T$ within $0.2\%$ of the optimum \cite{Evered2023}.
Via gradient descent methods, one can tune the parameters $\delta, A, \omega,\phi, T$ to optimize gate fidelity. We set the total gate time to $T= 7.612 / \Omega$. For our analysis we perform a two-step optimization. The first step uses SciPy's L-BFGS-B minimizer, which uses numerical gradients to walk to a local minimum. Due to the structure of the landscape, this optimizer can get stuck in local minima. To resolve this, stage 2 initializes the optimizer from different starting points, keeping the best result seen. We combine fully random initializations with initializations drawn from a Gaussian distribution around the best-found gate parameters. This does not guarantee full convergence to global maxima, but we heuristically find that the gates found by this protocol result in parameters that match the expected behavior of TO gates. 

\rsub{For the optimizations presented in the main body of this letter, we use Gaussian kicks with a distribution of size $0.3\times$ the best-found parameters (i.e. a Gaussian of size 0.3 and center 1).}
\rsub{\subsection{Almost Resonant Mudulated Driving (ARMD)}}

\rsub{The almost-resonant modulated-driving (ARMD) gate \cite{fan2024fast} implements a CZ operation using independently modulated Rabi frequencies on the control and target atoms. Each waveform Rabi frequency is represented by a truncated Fourier series,
\begin{equation}
    \Omega_{j}(t)
    =
    \frac{1}{2N+1}
    \left[
    a_{0}^{(j)}
    +
    2\sum_{n=1}^{N}
    a_{n}^{(j)}
    \cos\!\left(\frac{2\pi n t}{T}\right)
    \right],
    \qquad j\in\{c,t\},
\end{equation}
where $N=5$ and $T=0.32~\mu\mathrm{s}$ for the rescaled Rabi frequency. For the simulations presented here, the Fourier coefficients after rescaling are
\begin{align}
\bm{a}^{(c)}
&=
\left(
68.88011384,
-28.76983280,
-10.21344717,
1.62006403,
3.27143365,
-0.35218783
\right),\\
\bm{a}^{(t)}
&=
\left(
68.88011384,
-4.64105301,
-15.65279260,
-8.28032729,
-3.91319815,
-1.95659908
\right).
\end{align}
The two Rabi frequencies are rescaled using a common normalization such that
\begin{equation}
    \max_{c,t}\frac{|\Omega_j(t)|}{2\pi}
    =10~\mathrm{MHz},
\end{equation}
which preserves their relative amplitudes. }

\rsub{The Rabi frequencies on the control and target atoms, as well as the gate performance under one- and two-eigenstate models is shown in Fig. \ref{fig:armd}}

\begin{figure}
    \centering
    \rsub{\includegraphics[width=0.8\linewidth]{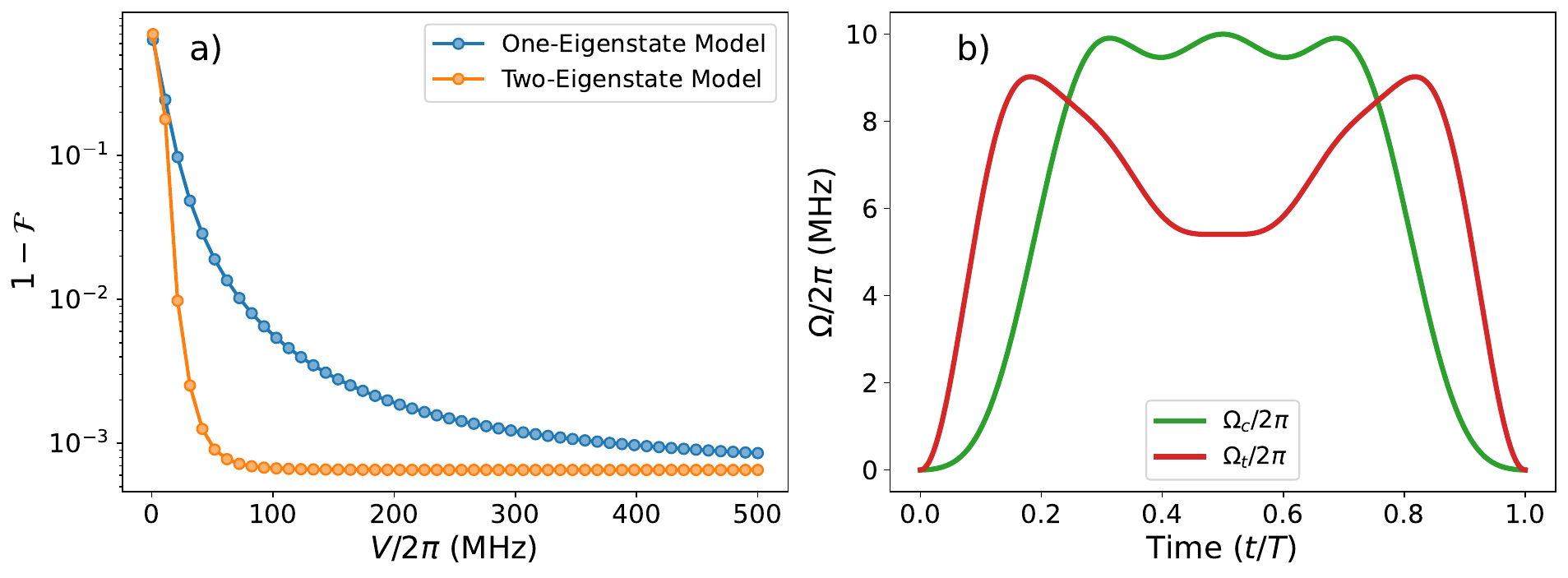}
    \caption{\textbf{(a)} Performance of the ARMD gate under one- and two-eigenstate models. \textbf{(b)} Time-varying Rabi rates used for the analysis in panel (a).}
    \label{fig:armd}}
\end{figure}

\section{Gate Error Scaling}
Here, we use the standard rank-one $\pi-2\pi-\pi$ gate as an example to show the reduced errors when considering the two-eigenstate model. 
\subsection{One-Eigenstate Model}
After the first $\pi$ pulse, the map 
$\ket{11}\mapsto\ket{a1}$ is realized.
The $2\pi$ pulse causes off-resonant excitation to the doubly-populated Rydberg state $\ket{ab}$, due to the fact that $\ket{ab}$  has finite detuning $V$. The drive Hamiltonian for this $2\pi$ pulse, in the basis of $\{\ket{a1}, \ket{ab}\}$ is 

\begin{equation}
    H_{2\pi}=\begin{pmatrix}
        0 & \Omega/2\\
        \Omega/2 & V
    \end{pmatrix}.
\end{equation}
Perturbation theory at second order then gives the energy shift on $\ket{a1}$ from this drive:
\begin{equation}
    \Delta E_{\ket{a1}}=\frac{|\bra{a1}H_{2\pi}\ket{ab}|^2}{E_{\ket{a1}}-E_{\ket{ab}}}=\frac{\left(\Omega/2\right)^2}{0-V}=-\frac{\Omega^2}{4V}
\end{equation}
This energy shift causes a phase
\begin{equation}
    \phi=-\Delta E_{\ket{a1}}\times t_{2\pi}=\frac{\Omega^2}{4V}\times\frac{2\pi}{\Omega}=\frac{\pi\Omega}{2V}.
\end{equation}
This phase limits gate fidelity and scales as $\frac{\Omega}{V}$.
\subsection{Two-Eigenstate Model}
In the two-eigenstate model, and in the basis $\{\ket{a1},\ket{ab},\ket{\alpha\beta}\}$, the rank-one drive Hamiltonian for the $2\pi$ pulse is
\begin{equation}
    H_{2\pi}=\begin{pmatrix}
        0 & \Omega/2 & 0\\
        \Omega/2 & 0 & V\\
        0 & V & 0
    \end{pmatrix}.
\end{equation}
The Hamiltonian has eigenvalues $\{0, \pm\Lambda\}$ with $\Lambda = \sqrt{V^2 + \Omega^2/4}$, and (non-normalized) eigenvectors
\begin{align}
\ket{\phi_0} &= \ket{a1} + -\epsilon\ket{\alpha\beta}, \\
\ket{\phi_\pm} &=\epsilon\ket{a1}\pm\tfrac{\Lambda}{V}\ket{ab}+\ket{\alpha\beta},
\end{align}
where $\epsilon \equiv \Omega/(2V)$. Decomposing the initial state $\ket{\psi(0)} = \ket{a1}$ into this eigenbasis, evolving under $H_{2\pi}$ for time $t$, and projecting back onto $\ket{a1}$ yields
\begin{equation}
\langle a1|\psi(t)\rangle = \frac{1 + \epsilon^2\cos(\Lambda t)}{1+\epsilon^2}.
\end{equation}
This quantity is real at all times, with an amplitude modulation at frequency $\Lambda \approx V$ with magnitude $\epsilon^2 \approx \Omega^2/(4V^2)$. The coherent error for the two-eigenstate model therefore scales quadratically in $\Omega/V$, compared to the linear scaling for the one-eigenstate model, thereby suppressing errors.

\section{Imperfect F\"orster Resonances}
We consider gate performance in the two-eigenstate model, allowing $o_1=1-o_2$ to deviate from the nominal value of $o_1=1/2$ in a perfect F\"orster resonance. Due to finite coupling between other non-target or exchange states, real atomic systems rarely exhibit perfect resonances. Such systems can be tuned into resonance, though, by applying local fields. In this section, we evaluate the sensitivity of the increase in gate performance to the symmetry of the pair potential. We use the same three gate forms analyzed in this letter: ARP, $\pi-2\pi-\pi$, and TO, shown in Fig \ref{fig:imperfect_resonances}. Gate fidelity monotonically increases with the pair state symmetry. For the TO gate, the gate was optimized under the two-eigenstate model with $o_1=o_2=1/2$ and evaluated with a scan of $o_1$. Due to the optimization procedure, higher gate fidelities than those presented can be achieved with imperfect resonances, but require different pulses. We present the fidelity variation for the gate optimized for a perfect F\"orster resonance.

\begin{figure*}
    \centering
    \includegraphics[width=\linewidth]{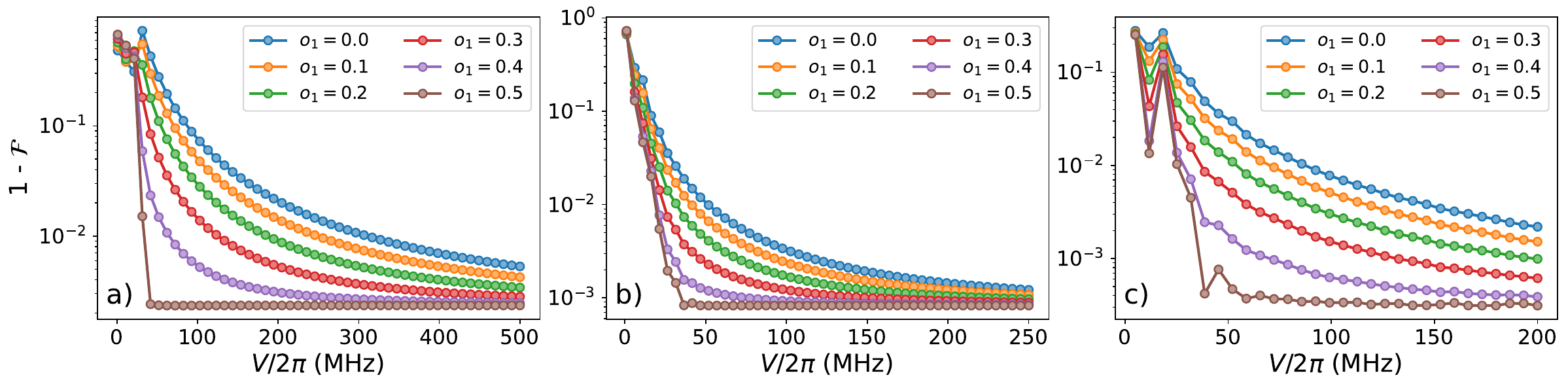}
    \caption{Gate performance with imperfect F\"orster resonances. \textbf{(a)} ARP gate. \textbf{(b)} $\pi-2\pi-\pi$ gate. \textbf{(c)} TO gate.}
    \label{fig:imperfect_resonances}
\end{figure*}

\rsub{\subsection{Atomic States}}
\rsub{We consider now real atomic systems, whose eigenstates are determined by a direct diagonalization of the dipole-dipole interaction Hamiltonian, summed over many states. We define a figure-of-merit, the overlap of the dominant two eigenstates to each of the two channels in the F\"orster resonance. First, we recall the eigenstates of the perfect resonance}
\rsubmath{
\begin{align}
    |e_{1}\rangle = \frac{1}{\sqrt{2}}\left(|ab\rangle +|\alpha \beta\rangle \right),\quad |e_{2}\rangle = \frac{1}{\sqrt{2}}\left(|ab\rangle -|\alpha \beta\rangle \right).
\end{align}
And for the general real atomic eigenstates
\begin{align}
    |\psi_{1}\rangle = c_{11}|ab\rangle +c_{12}|\alpha \beta\rangle+\cdots ,\quad |\psi_{2}\rangle = c_{21}|ab\rangle + c_{22}|\alpha \beta\rangle +\cdots.
\end{align}
First we define their populations inside the selected two-state subspace
\begin{align}
    L_{1} = |c_{11}|^{2} + |c_{12}|^{2},\quad L_{2} = |c_{21}|^{2} + |c_{22}|^{2}
\end{align}
And also define the normalized coherence of each eigenstate
\begin{align}
    q_{1} = \frac{2c_{11}c_{12}}{L_{1}},\quad  q_{2} = \frac{2c_{21}c_{22}}{L_{2}}
\end{align}}

\rsub{Then we can define the strength of a F\"orster resonance using the function}
\rsubmath{\begin{equation}
    M= L_{1}L_{2}\frac{\abs{q_{1} - q_{2}}^2}{4}.
\end{equation}}
\rsub{$M$ is therefore bounded on the interval $[0,1]$, where a perfect single-channel interaction has $M=0$ and a perfect F\"orster resonance has $M=1$. As an example, we show the pair state $\ket{\text{Rb } 55 p_{3/2},m_j=3/2}\otimes \ket{\text{Cs } 50 p_{3/2},m_j=3/2}$, which exhibits a strong F\"orster resonance in Fig. \ref{fig:real-forster-resonance}. Pair states were calculated with the Alkali Rydberg Calculator \cite{Sibalic2017}, with a Hamiltonian generated from the following parameters: $\Delta n =\Delta \ell = 5, \Delta E_{\rm max}= 25$ GHz, in a $\theta=\phi=0$ geometry.}

\begin{figure}
    \centering
    \rsub{
    \includegraphics[width=0.95\linewidth]{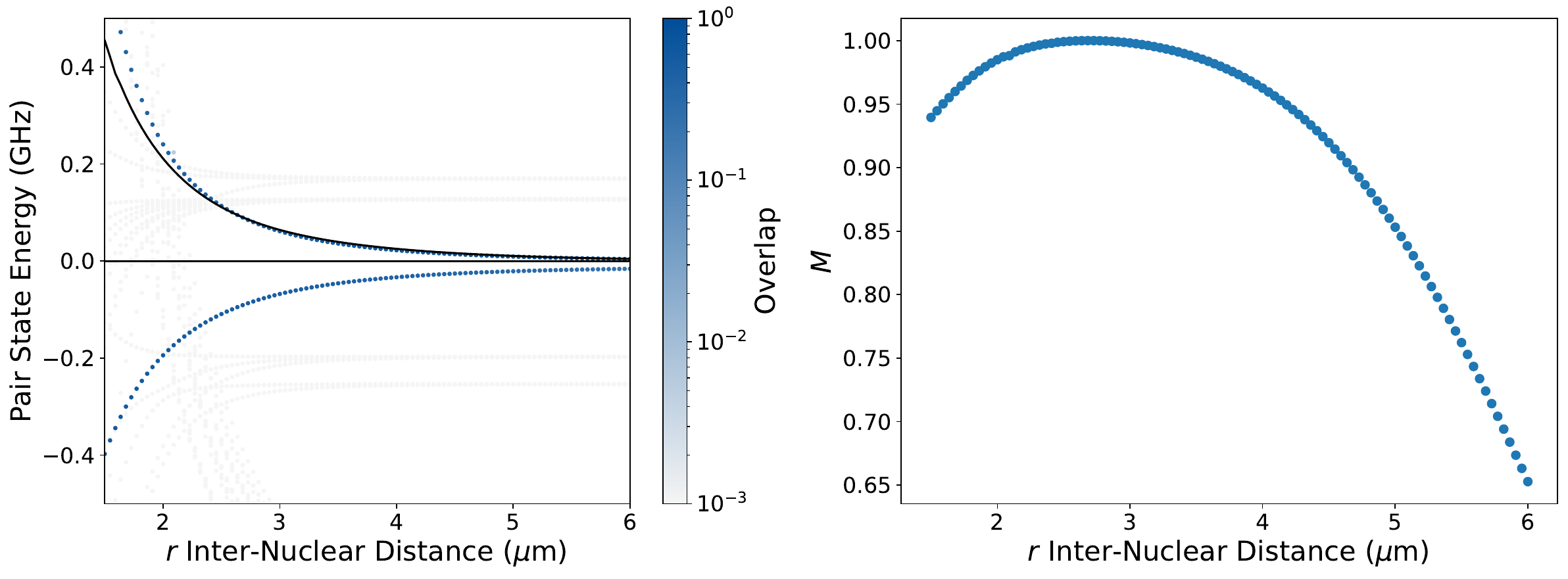}
    \caption{An inter-species F\"orster resonance with the states $\ket{\text{Rb } 55 p_{3/2},m_j=3/2}\otimes \ket{\text{Cs } 50 p_{3/2},m_j=3/2}$. \textbf{(a)} The atomic pair state, found by a diagonalization of the atom-atom interaction Hamiltonian. The single-channel interaction energy $V$ is shown in black line. \textbf{(b)} The strength of the resonance, characterized by the figure-of-merit $M$ for the same pair state.}
    \label{fig:real-forster-resonance}
    }
\end{figure}

\rsub{There exist other pair states that naturally exhibit high F\"orster character, but we only show this one for simplicity. Additionally, we mention that states near a resonance can be tuned further into resonance through the application of external electric and magnetic fields.}

\rsub{\section{Phase Accumulation In Interaction Gate}}

\rsub{Here we decompose the phase accumulation of the rank-two $\pi$--gap--$\pi$ protocol into dynamical and geometric contributions, and show that the suppression of coherent errors follows from an exact, symmetry-protected cancellation of AC-Stark shifts at a perfect F\"orster resonance.}

\rsub{The gate is diagonal in the basis
$\{\ket{00},\ket{11},\ket{\Psi^+},\ket{\Psi^-}\}$ with $\ket{\Psi^\pm}=(\ket{01}\pm\ket{10})/\sqrt{2}$. Each of these branches undergoes a cyclic evolution $\ket{\psi(T)}=e^{i\varphi_{\rm tot}}\ket{\psi(0)}$, so the total phase of
each branch splits into a dynamical and a geometric part,}
\rsubmath{
\begin{equation}
\varphi_{\rm tot}=\varphi_{\rm dyn}+\varphi_{\rm geo},
\qquad
\varphi_{\rm dyn}=-\int_0^T dt\,\bra{\psi(t)}H(t)\ket{\psi(t)}.
\end{equation}}

\rsub{The excitation pulse maps to the eigenstates,}
\rsubmath{\begin{equation}
\ket{\Psi^+}\mapsto -\ket{e_1}=-\tfrac{1}{\sqrt2}\left(\ket{ab}+\ket{\alpha\beta}\right),
\qquad
\ket{\Psi^-}\mapsto -\ket{e_2}=-\tfrac{1}{\sqrt2}\left(\ket{ab}-\ket{\alpha\beta}\right),
\end{equation}}
\rsub{while $\ket{00},\ket{11}$ map onto the
stationary states $\ket{a\beta},\ket{\alpha b}$ with $\bra{a\beta}H_{\rm int}\ket{a\beta}=0$.
The decomposition in the limit $\Omega\gg V$ is then:}
\rsub{
\begin{itemize}
\item[--] \emph{Pulses:} for resonant Rabi driving, $\expval{H_{\rm drive}}=0$, so the pulses contribute \emph{zero dynamical phase}. The pulse
contribution per two-atom branch is therefore geometric, but since each branch undergoes the same transformation (i.e. an excitation and a de-excitation), the phase on each branch is the same, and can thus be ignored. 
\item[--] \emph{Gap:} the exchange interaction itself is a dynamical phase; the stationary branches $\ket{00},\ket{11}$ accumulate exactly zero phase.
\end{itemize}
}
\rsub{The entangling phase is thus purely dynamical.}

%\bibliography{bib_files/atomic, bib_files/optics, bib_files/qc_refs, bib_files/rydberg, bib_files/saffman_refs, bib_files/this_paper}

% \end{document}

\end{document}